\documentclass[a4paper,10pt,reqno]{amsart}

\usepackage[numbers]{natbib}
\usepackage{amsmath,amssymb}
\usepackage[foot]{amsaddr}
\usepackage{caption}
\usepackage{mathtools}
\usepackage[latin1]{inputenc}
\usepackage{mathrsfs}
\usepackage{enumerate}
\usepackage{textcmds}
\usepackage{graphicx,epstopdf}
\usepackage[caption=false]{subfig}
\usepackage[pdfusetitle]{hyperref}
\usepackage[capitalize]{cleveref}

\newtheorem{Theorem}{Theorem}[section]
\newtheorem{Lemma}[Theorem]{Lemma}
\newtheorem{Corollary}[Theorem]{Corollary}

\newcommand{\D}[2]{\frac{\mathrm{d}#1}{\mathrm{d}#2}}
\newcommand{\T}{^{\operatorname{T}}}
\newcommand{\Osym}{\text{\usefont{OMS}{cmsy}{m}{n}O}}
\newcommand{\N}{\mathbb{N}}
\newcommand{\Z}{\mathbb{Z}}
\newcommand{\R}{\mathbb{R}}
\newcommand{\C}{\mathbb{C}}
\newcommand{\Cz}{\C^{2}}
\newcommand{\MzC}{\mathrm{M}_{2}(\C)}
\newcommand{\Poly}{\mathscr{P}}
\newcommand{\E}{\mathrm{e}}
\newcommand{\I}{\mathrm{i}}
\newcommand{\re}{\operatorname{Re}}
\newcommand{\im}{\operatorname{Im}}
\newlength{\ml}
\newcommand{\ms}{\hspace*{\ml}}

\title[Eigenvalues of the Coulomb spheroidal wave equation]{Computation of the eigenvalues for the\\ angular and Coulomb spheroidal wave equation}
\author{Harald Schmid}
\email{h.schmid@oth-aw.de}
\address{University of Applied Sciences Amberg-Weiden, Amberg, Germany}

\begin{document}

\begin{abstract}
In this paper we study the eigenvalues of the angular spheroidal wave equation and its generalization, the Coulomb spheroidal wave equation. An associated differential system and a formula for the connection coefficients between the various Floquet solutions give rise to an entire function whose zeros are exactly the eigenvalues of the Coulomb spheroidal wave equation. This entire function can be calculated by means of a recurrence formula with arbitrary accuracy and low computational cost. Finally, one obtains an easy-to-use method for computing spheroidal eigenvalues and the corresponding eigenfunctions.
\end{abstract}

\maketitle

\section{Introduction}

The angular spheroidal wave equation (or ASWE for short)
\begin{equation} \label{ASWE}
\D{}{x}\left((1-x^2)\D{}{x}w(x)\right) + \left(\lambda + \gamma^2(1-x^2) - \frac{\mu^2}{1-x^2}\right)w(x) = 0,\quad -1<x<1
\end{equation}
appears in many fields of physics and engineering like quantum mechanics, electromagnetism, signal processing etc. In particular, if $\mu=m$ is an integer and $\gamma^2$ is real, then the separation of the Helmholtz equation in prolate ($\gamma^2>0$) or oblate ($\gamma^2<0$) spheroidal coordinates results in a second order ODE of the form \eqref{ASWE}. The ASWE is a special case of the generalized spheroidal wave equation (GSWE)
\begin{equation} \label{GSWE}
\D{}{x}\left((1-x^2)\D{}{x}w(x)\right) + \left(\lambda + \beta x + \gamma^2(1-x^2) - \frac{\mu^2+\alpha^2+2\alpha\mu x}{1-x^2}\right)w(x) = 0
\end{equation}
which in turn is equivalent to the confluent Heun differential equation (see \cite[Section 3.1.2]{SL:2000}). If we set $\alpha=0$, then we get the Coulomb spheroidal wave equation (CSWE)
\begin{equation} \label{CSWE}
\D{}{x}\left((1-x^2)\D{}{x}w(x)\right) + \left(\lambda + \beta x + \gamma^2(1-x^2) - \frac{\mu^2}{1-x^2}\right)w(x) = 0
\end{equation}
The numbers $\lambda\in\C$ for which \eqref{CSWE} has a nontrivial bounded solution $w(x)$ on $(-1,1)$ are the eigenvalues of the CSWE, and the corresponding eigenfunctions $w(x)$ are the so-called Coulomb spheroidal wave functions. They provide, for example, exact wave functions for a one-electron diatomic molecule with fixed nuclei (see \cite[Chapter 9]{Falloon:2001}), and they also arise in gravitational physics (see e.g. \cite{Leaver:1986}). In this paper, we are mainly concerned with equation \eqref{CSWE}, whereas the results we obtain are obviously applicable to the ``ordinary'' spheroidal wave equation \eqref{ASWE} as well.

Over the years, various approaches for calculating the eigenvalues of the ASWE and CSWE have been developed. One of these standard methods (see \cite{MS:1954} or \cite{FAW:2003}) is based on a series expansion by means of associated Legendre functions: A three term recurrence relation for the coefficients of this expansion results in a transcendental equation involving continued fractions, whose roots are the spheroidal eigenvalues. The numerical values for $\lambda$ can be computed by an iterative method (cf. \cite{Flammer:1957}) or can be approximated by the eigenvalues of an associated symmetric tridiagonal matrix (see e.g. \cite{Hodge:1970}). There is also a different approach, a type of shooting method, where the regular Floquet solutions at $x=\pm 1$ are smoothly matched at $x=0$; the eigenvalues of the ASWE coincide with the zeros of the corresponding Wronskian, cf. \cite{Skoro:2015}.

The strategy that we use in the present paper is also based on matching certain Floquet solutions, but not for the CSWE \eqref{CSWE} itself. Instead, we study an associated linear $2\times 2$ differential system of the type $\eta'(z) = (\frac{1}{z}A_0 + \frac{1}{z-1}A_1 + C)\eta(z)$ with two regular-singular points at $z=0$ and $z=1$. The structure of this system is similar to that of the Chandrasekhar-Page angular equation. We can therefore determine the eigenvalues using an approach analogous to that in \cite[Lemma 3]{BSW:2004}, proceeding as follows: In a neighborhood of $z=0$, the $2\times 2$ system has a fundamental set which consists of a holomorphic solution $\eta_0(z)=\sum_{k=0}^\infty z^k d_k$ and a solution which behaves like $z^{-\mu-1}e_1$ (here and in the following, $\{e_1,e_2\}$ denotes the standard basis of $\Cz$). In addition, there is a second set of fundamental solutions $\eta_1(z)$ and $\eta_2(z)$, where $\eta_1(z)\sim (z-1)^{-1}e_1$ as $z\to 1$ and $\eta_2(z)$ is bounded near $z=1$. The solutions $\eta_0(z)$ and $\eta_1(z)$, $\eta_2(z)$ are related by a linear combination $\eta_0(z) = c_1\eta_1(z) + c_2\eta_2(z)$ with some {connection coefficients $c_1$ and $c_2$, which depend holomorphically on $\lambda$. We will prove that \eqref{CSWE} has a nontrivial solution $w(x)$ which is bounded on $(-1,1)$ if and only if the associated system has a nontrivial solution $\eta(z)$ which is holomorphic at $z=0$ \emph{and} $z=1$. This solution must be a constant multiple of $\eta_0(z)$ and $\eta_2(z)$, which means that $c_1(\lambda)=0$. Finally we use some results of R. Schäfke and D. Schmidt from \cite{Schaefke:1980}, \cite{SS:1980} to prove that the connection coefficient $c_1$ is the limit of a sequence generated by the series coefficients $d_k\in\Cz$, and it is even possible to specify the order of convergence. The vectors $d_k$ can be computed with a relatively simple recursion formula, and thus also $c_1$ can be evaluated by a straightforward algorithm with arbitrary accuracy. Subsequently, only the zeros of $c_1=c_1(\lambda)$ have to be determined, and for this purpose one can use, for example, the secant method.

Since the mathematical background is somewhat tedious and rather technical, we first present the main result with some numerical examples in \cref{sec:Main} before we prove the main result in \cref{sec:Proof}. Finally, in \cref{sec:GSWE} we briefly outline how to obtain a corresponding algorithm for the eigenvalues and eigenfunctions of the generalized spheroidal wave equation \eqref{GSWE}.

\section{Main theorem and numerical results}
\label{sec:Main}

In the following we may assume without loss of generality that $\re\mu\geq 0$ holds, since the differential equation \eqref{CSWE} does not change when $\mu$ is replaced by $-\mu$. Now, for fixed values $\beta,\gamma\in\C$ and a given number $t\in\C$, we define a sequence of vectors $u_k,d_k\in\Cz$ by means of a recurrence relation
\begin{equation}\label{Recurrence}
\begin{split}
u_k & := \begin{pmatrix} 
0 &  \frac{t-\beta}{k}-\frac{2t}{k+\mu+1} \\[1ex] 
0 & -\frac{\mu+1}{k} \end{pmatrix} d_{k-1} - \begin{pmatrix} 
 \frac{t-\beta}{k(k+\mu+1)} & \frac{4\gamma^2}{k+\mu+1} \\[1ex]
-\frac{1}{k} & 0 \end{pmatrix} u_{k-1} \\
d_k & := d_{k-1}+u_k\quad\mbox{for}\quad k=1,2,3,\ldots\quad\mbox{with}\quad
u_0 = d_0 := \begin{pmatrix} \frac{\beta-t}{\mu+1} \\[1ex] 1 \end{pmatrix}
\end{split}\end{equation}

The following theorem is the main result; it describes how the sequence of vectors $(d_k)_{k=0}^\infty$ can be used to calculate the eigenvalues and eigenfunctions of the Coulomb spheroidal wave equation. 

\begin{Theorem} \label{thm:MainRes}
Let $\mu,\beta,\gamma\in\C$ be fixed, and suppose that either $\re\mu>0$ or $\mu=0$ holds. Further, let
\begin{equation*}
\Theta_k := \langle \vartheta,d_k\rangle = \vartheta\T d_k\quad\mbox{for}\quad k=1,2,3,\ldots
\end{equation*}
be the scalar product of the vector
\begin{equation} \label{theta}
\vartheta := \begin{pmatrix} \ms 1 \\[1ex] -\frac{\beta+t}{\mu+1} \end{pmatrix}
\end{equation}
and $d_k$ given by \eqref{Recurrence}. Then $\Theta_k=\Theta_k(t)$ is a polynomial of degree $k+1$ in $t\in\C$. Moreover, the limit 
\begin{equation*}
\Theta(t) := \lim_{k\to\infty}\Theta_k(t) 
\end{equation*}
exists for each $t\in\C$, and it has the following properties: 
\begin{enumerate}[\upshape (a)]
\item $\Theta:\C\longrightarrow\C$ is an entire function;
\item $\Theta_k(t) = \Theta(t) + \Osym(k^{\varepsilon-\mu-2})$ as $k\to\infty$ for each $t\in\C$ with arbitrary small $\varepsilon>0$;
\item $\lambda\in\C$ is an eigenvalue of the CSWE \eqref{CSWE} if and only if $t = \lambda-\mu(\mu+1)$ is a zero of $\Theta$. In this case
\begin{equation*}
w(x) := \left(\frac{1+x}{1-x}\right)^{\mu/2}\sum_{k=0}^\infty \tfrac{1}{2^k}e_2\T d_k(1+x)^k
\end{equation*}
is an eigenfunction corresponding to $\lambda$, i.e., a nontrivial solution of \eqref{CSWE} which is bounded on $(-1,1)$; it behaves like $w(x) = (1+x)^{\mu/2}\left(1 + o(1)\right)$ for $x\to -1$ and like $w(x) = (1-x)^{\mu/2}\left(c + o(1)\right)$ for $x\to 1$ with some constant $c\in\C$.
\item In the special case $\beta=\gamma=0$ the function $\Theta(t)$ becomes
\begin{equation*}
\Theta(t) 
= \frac{\cos\left((\tau-\mu)\pi\right)\Gamma(\mu+1)^2\Gamma(\tau+\frac{1}{2}-\mu)}{\pi\Gamma(\tau+\frac{1}{2}+\mu)} = \frac{\Gamma(\mu+1)^2}{\Gamma(\mu+\frac{1}{2}-\tau)\Gamma(\mu+\frac{1}{2}+\tau)}
\end{equation*}
where $\tau := \sqrt{t+(\mu+\tfrac{1}{2})^2}$ and $-\frac{\pi}{2}<\arg(\tau)\leq\frac{\pi}{2}$.
\end{enumerate}
\end{Theorem}

The proof of this theorem can be found in the next section. Before we address the numerical computation of the spheroidal eigenvalues, let us first have a look at the special case $\beta=\gamma=0$, in which \eqref{CSWE} reduces to the associated Legendre differential equation. Since the reciprocal Gamma function $\frac{1}{\Gamma(z)}$ is an entire function with simple zeros at $z=0,-1,-2,-3,\ldots$, the zeros of the function 
\begin{equation*}
\frac{\Gamma(\mu+1)^2}{\Gamma(\mu+\frac{1}{2}-\tau)\Gamma(\mu+\frac{1}{2}+\tau)}
\end{equation*}
in the sector $-\frac{\pi}{2}<\arg(\tau)\leq\frac{\pi}{2}$ are given by $\tau_n=n+\mu+\frac{1}{2}$, where $n$ is an arbitrary non-negative integer. Hence, according to (d) in \cref{thm:MainRes}, the zeros of $\Theta(t)$ are located at $t_n = \tau_n^2-(\mu+\frac{1}{2})^2=n(n+2\mu+1)$, and from (c) it follows that the eigenvalues of the associated Legendre differential equation are determined by $\lambda_n = t_n+\mu(\mu+1) = n(n+2\mu+1)+\mu(\mu+1)$, which coincides with the well-known formula
\begin{equation*}
\lambda_n = (n+\mu)(n+\mu+1),\quad n=0,1,2,3,\ldots
\end{equation*}

Now we return to the general case $\beta,\gamma\in\C$. A closer view on the recursion formula \eqref{Recurrence} shows that the first components of the vectors $d_k$ are not required for the calculation of $\Theta$. Using the entries of the vectors
\begin{equation*}
u_k = \begin{pmatrix} a_k \\[1ex] b_k \end{pmatrix},\quad
d_k = \begin{pmatrix} \ast \\[1ex] w_k \end{pmatrix}
\end{equation*}
we can deduce from \eqref{Recurrence} a more straightforward procedure for the computation of $\Theta(t)$.

\begin{Corollary} \label{cor:NumCom}
Suppose that either $\re\mu>0$ or $\mu=0$ holds. If we define 
\begin{equation*}
\begin{split}
a_k & := \frac{\beta-t}{k(k+\mu+1)}\,a_{k-1} - \frac{4\gamma^2}{k+\mu+1}\,b_{k-1} 
       + \frac{(t-\beta)(\mu+1)-(t+\beta)k}{k(k+\mu+1)}\,w_{k-1} \\
b_k & := \frac{1}{k}\,a_{k-1} - \frac{\mu+1}{k}\,w_{k-1},\quad w_k := b_k + w_{k-1},\quad
\Theta_k := \Theta_{k-1} + a_k - \frac{\beta+t}{\mu+1}\,b_k
\end{split}
\end{equation*}
for $t\in\C$ and $k=1,2,3,\ldots$ starting with $a_0 =\frac{\beta-t}{\mu+1}$, $b_0=w_0=1$, $\Theta_0 = -\frac{2t}{\mu+1}$, then
\begin{equation*}
\Theta_k = \Theta(t) + \Osym(k^{\varepsilon-\mu-2})
\end{equation*}
as $k\to\infty$ with arbitrary small $\varepsilon>0$. Moreover, $\lambda\in\C$ is an eigenvalue of the CSWE \eqref{CSWE} if and only if $t = \lambda-\mu(\mu+1)$ is a zero of $\Theta$,
and in this case the corresponding eigenfunctions are constant multiples of
\begin{equation*}
w(x) := \left(\frac{1+x}{1-x}\right)^{\mu/2}\sum_{k=0}^\infty\frac{w_k}{2^k}\,(1+x)^k
\end{equation*}
\end{Corollary}

According to \cref{cor:NumCom}, the eigenvalues $\lambda_n$ of \eqref{CSWE} are related to the zeros $t_n$ of the function $\Theta(t) = \Theta(\mu,\beta,\gamma;t)$ by means of a constant shift $\lambda_n=t_n+\mu(\mu+1)$. Thus, for fixed values $\beta,\gamma\in\C$ and $\mu\in\C$ with $\re\mu>0$ or $\mu=0$, we can now define an entire function
\begin{equation*}
\tilde\Theta(\lambda) := \Theta\big(\lambda-\mu(\mu+1)\big),\quad\lambda\in\C
\end{equation*}
such that the zeros $\lambda_n = \lambda_n(\mu,\beta,\gamma)$ of $\tilde\Theta(\lambda)$ are exactly the eigenvalues of the Coulomb spheroidal wave equation \eqref{CSWE}. 

\Cref{fig:Theta-3D} and \cref{fig:Theta-2D} illustrate the functions $\tilde\Theta(\lambda)$ for $\mu=0$ and $\mu=1$ with real parameters $-40\leq\lambda\leq 120$, $-80\leq\gamma^2\leq 80$ in the case $\beta=0$ along with their zero sets. These curves are the eigenvalues of the angular spheroidal wave equation \eqref{ASWE} for the specified parameters. Note that the top views in \cref{fig:Theta-2D} are consistent with the eigenvalue maps for $\lambda_n^0(\gamma^2)$ and $\lambda_n^1(\gamma^2)$ given by Meixner and Schäfke \cite[p. 236, figs. 13 and 14]{MS:1954}.

\begin{figure}[tbhp]
\centering
\subfloat[$\mu=0$]{\includegraphics{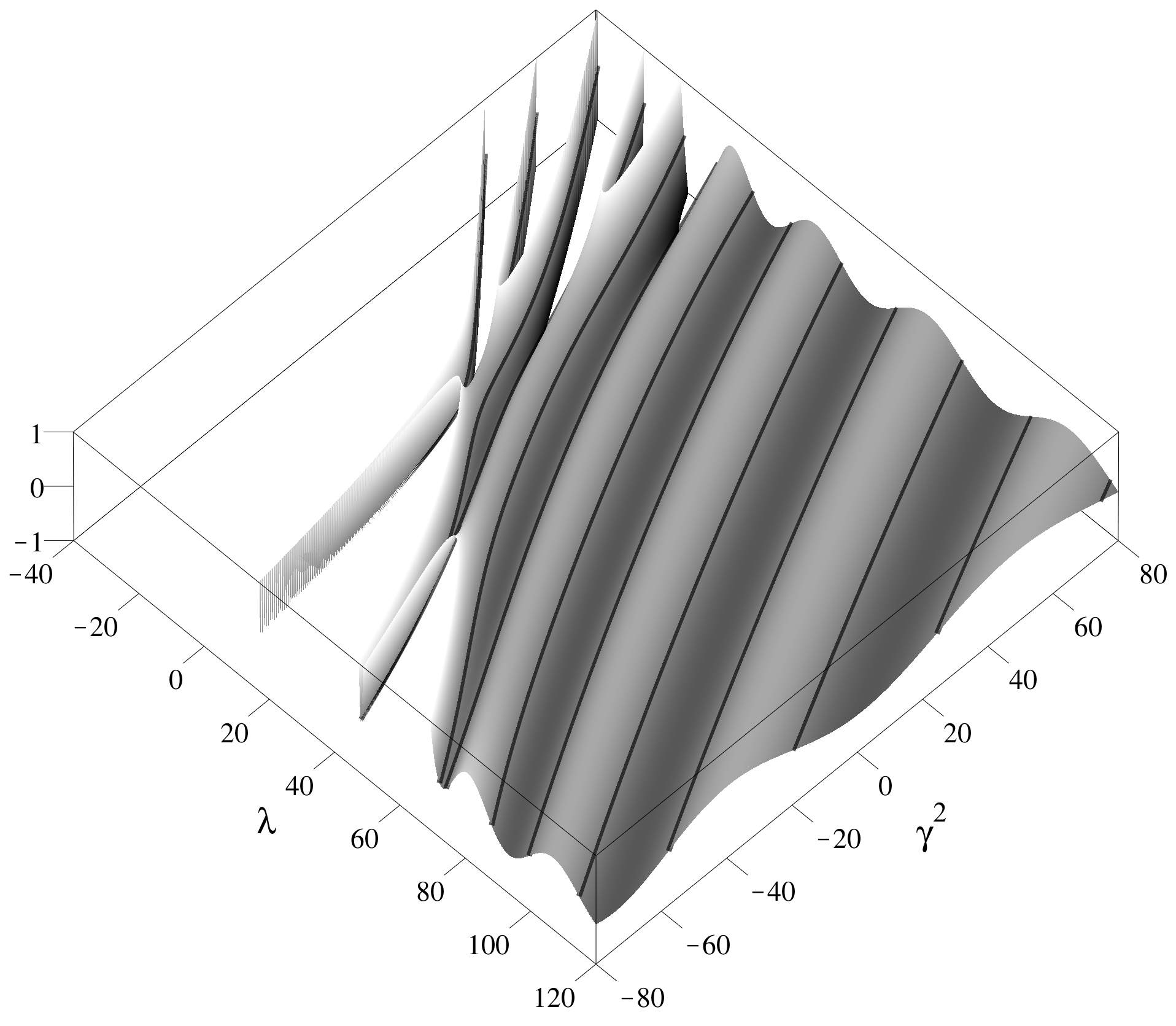}} \\
\subfloat[$\mu=1$]{\includegraphics{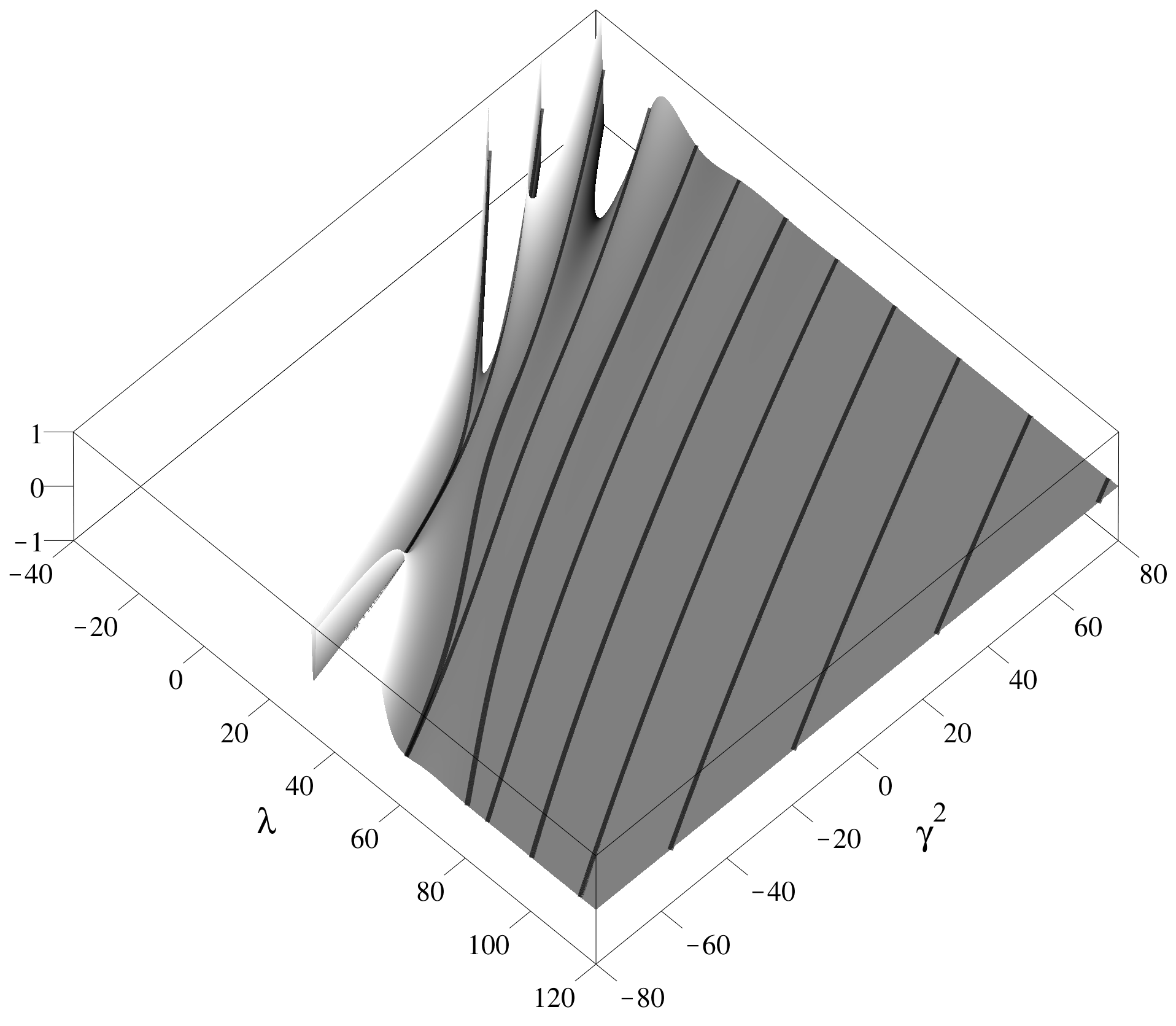}} 
\caption{The functions $\tilde\Theta(\lambda)$ for $\mu=0$ and $\mu=1$ in the case $\beta=0$. The contour lines are the zeros of $\tilde\Theta(\lambda)$, i.e., the eigenvalues of the angular spheroidal wave equation.}
\label{fig:Theta-3D}
\end{figure}

\begin{figure}[tbhp]
\centering
\subfloat[$\mu=0$]{\includegraphics{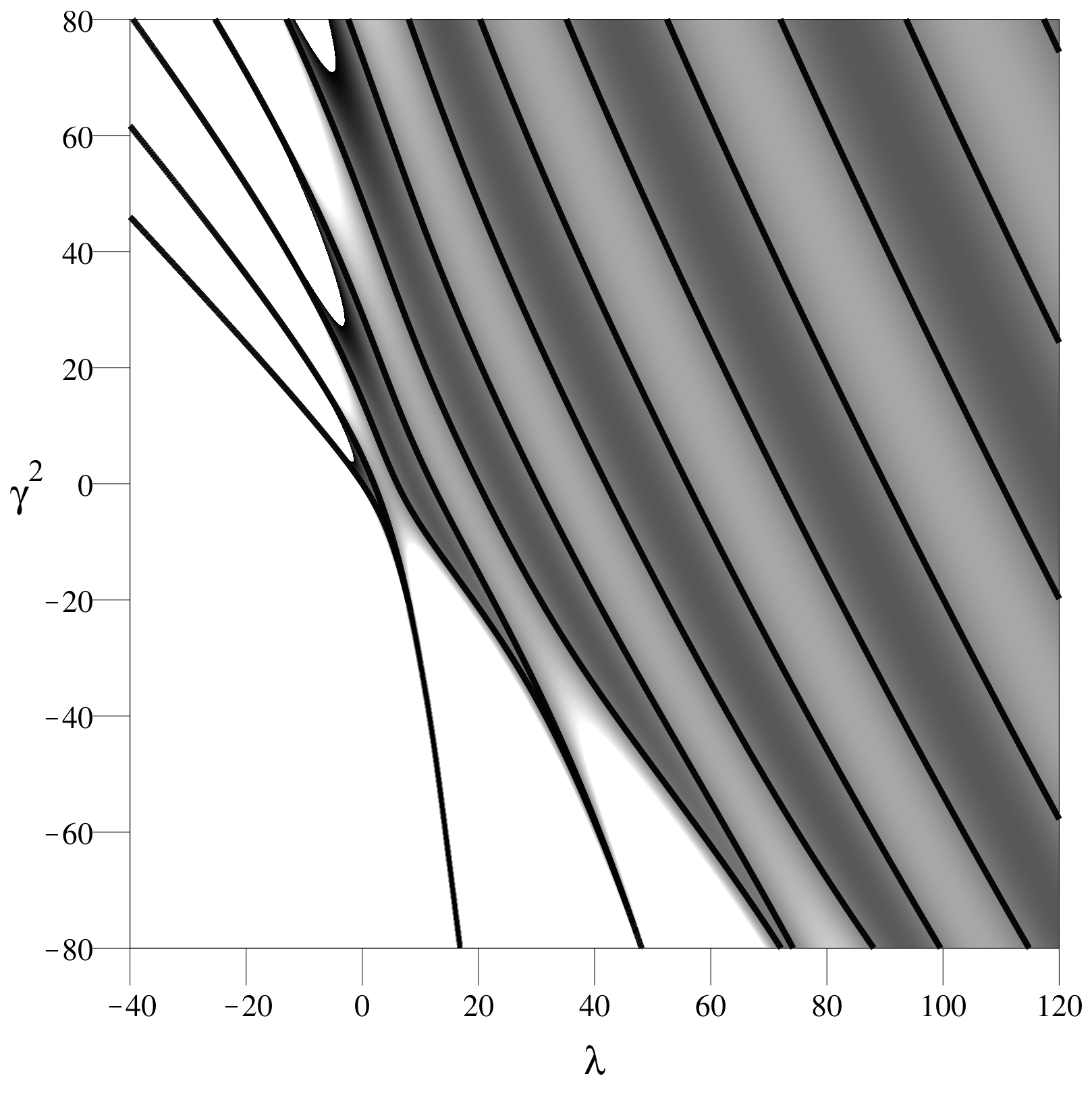}} \qquad 
\subfloat[$\mu=1$]{\includegraphics{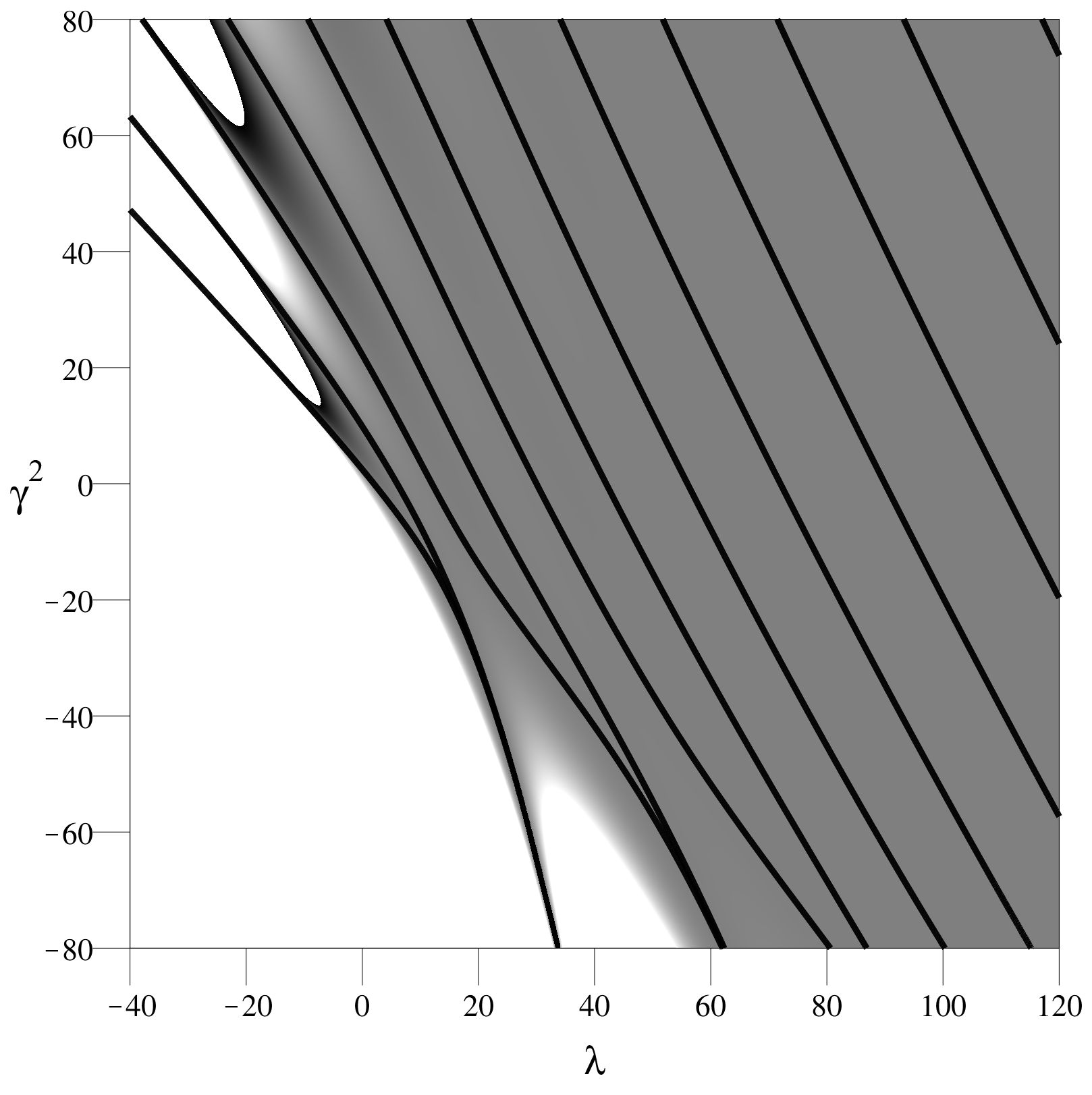}}
\caption{Eigenvalue maps for the the angular spheroidal wave equation with parameter values $\mu=0$ and $\mu=1$, produced by the level curves $\tilde\Theta(\lambda)=0$.}
\label{fig:Theta-2D}
\end{figure}

\Cref{fig:Theta-beta} shows the dependency of the eigenvalues $\lambda$ on the parameter $0\leq\beta\leq 90$ for fixed $\gamma=10\,\I$ (resp. $\gamma^2=-100$) as an example. The eigenvalue curves in this picture do not cross, as can be seen in the enlarged detail on the right. This phenomenon is known as ``avoided crossing''. It should be noted that, like in this example, when computing Coulomb spheroidal eigenvalues for some parameter $\beta$, one may always assume $\re(\beta)\geq 0$ without restriction, due to the following reason: If we substitute $-x$ for $x$ in \eqref{CSWE}, then we get the same CSWE except that $\beta$ is replaced by $-\beta$. Thus, at fixed $\mu$ and $\gamma$, the eigenvalues are identical for $\pm\beta$.

\begin{figure}[tbhp]
\centering
\includegraphics{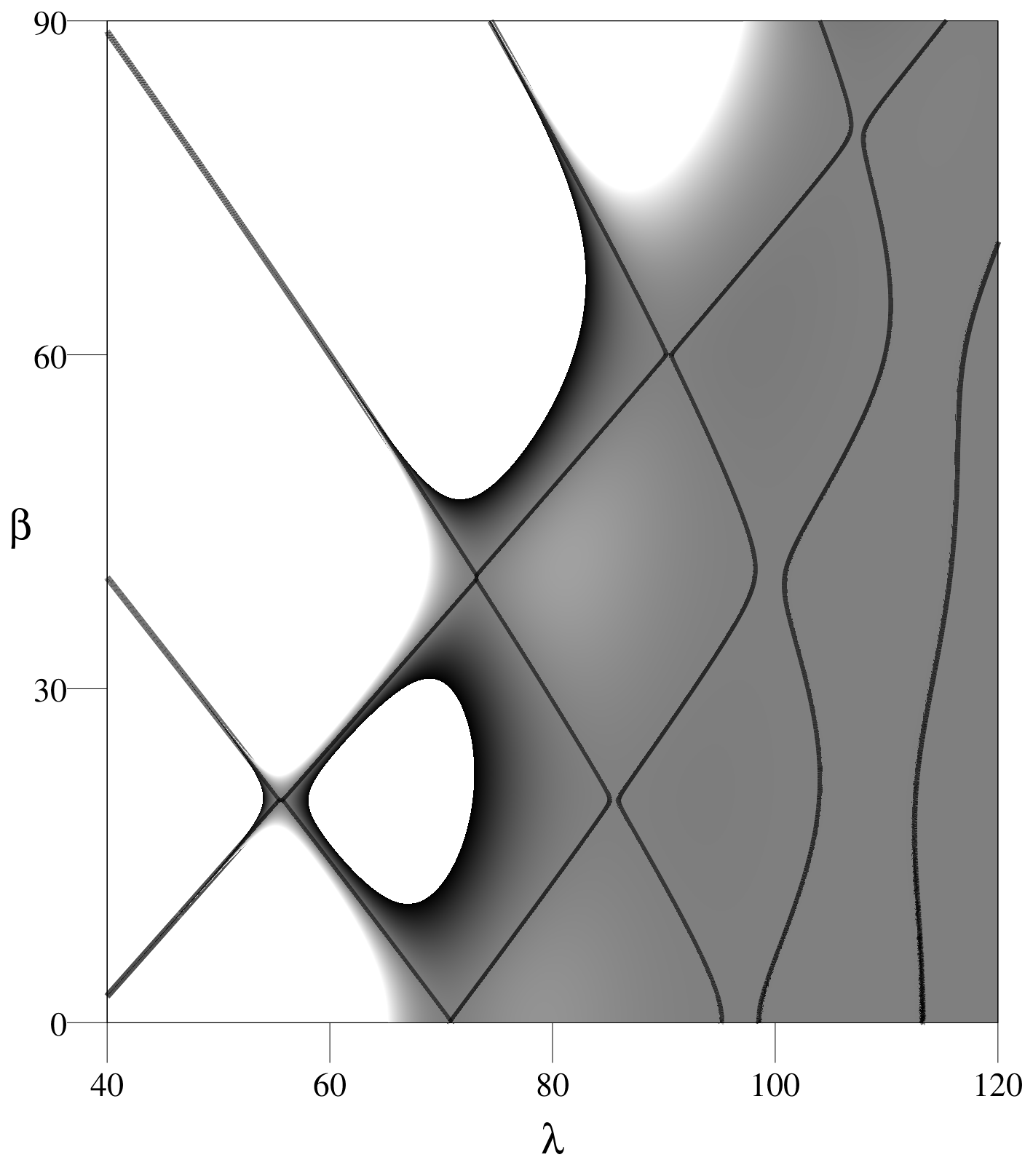}\qquad\includegraphics{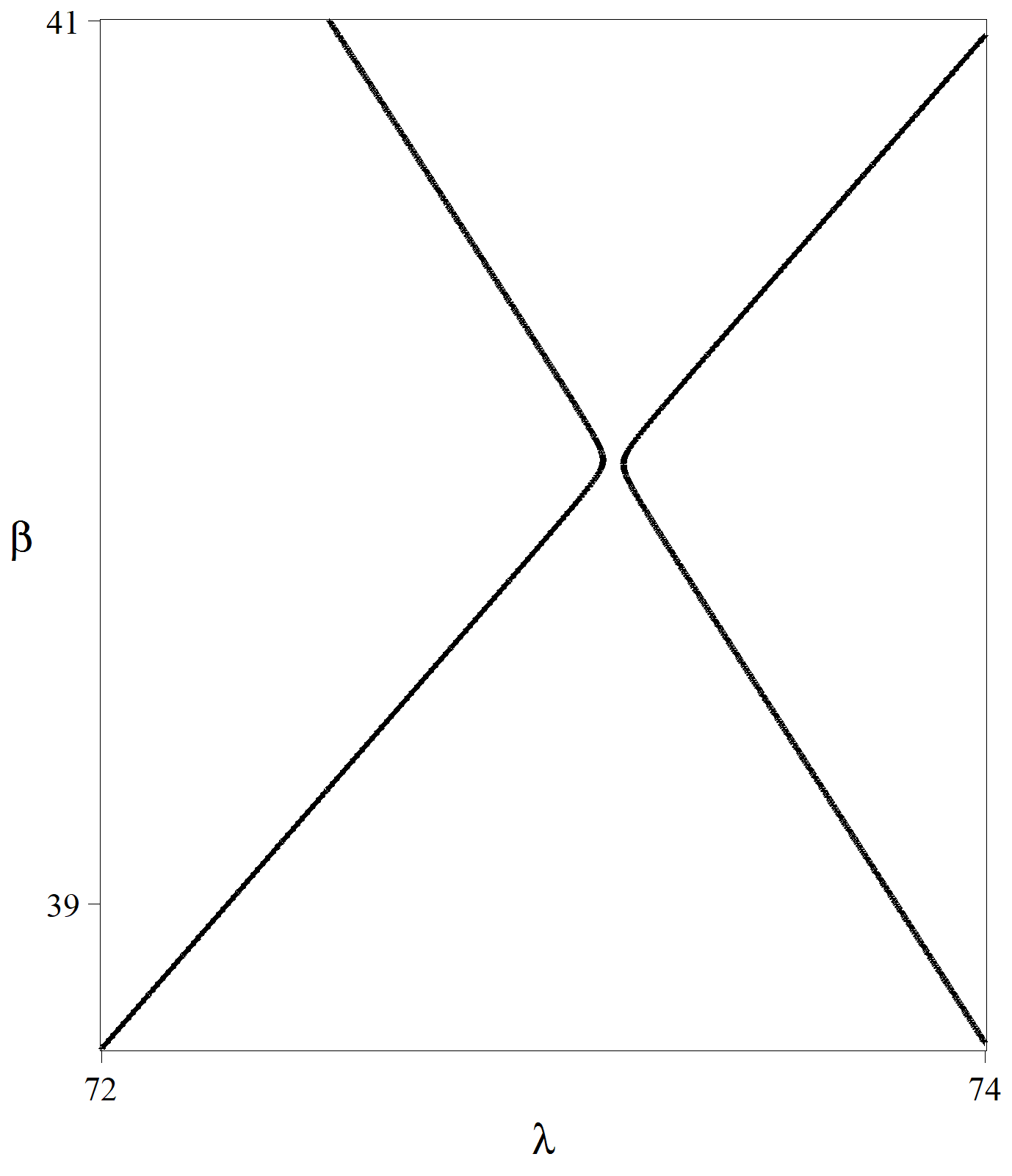}
\caption{The zeros of the function $\tilde\Theta(\lambda)$ for $\mu=1$ and $\gamma^2=-100$ in dependence of $\beta$.}
\label{fig:Theta-beta}
\end{figure}

As an example for an angular spheroidal wave equation ($\beta=0$) with complex parameters, we consider the cases $\mu = 2$ and $\mu=2+0.05\,\I$ for fixed $\gamma=5\,\I$. In \cref{fig:Theta-complex} the zeros of $\re\tilde\Theta(\lambda)$ are plotted as dashed lines and the zeros of $\im\tilde\Theta(\lambda)$ as solid lines. The intersection points of these curves are the zeros of the function $\tilde\Theta(\lambda)$ and hence the complex eigenvalues of the corresponding spheroidal wave equations. A numerical computation of the (complex) zeros provides the eigenvalues in \cref{tab:lambda2n}.

\begin{figure}[tbhp]
\centering
\subfloat[$\mu = 2$]{\includegraphics{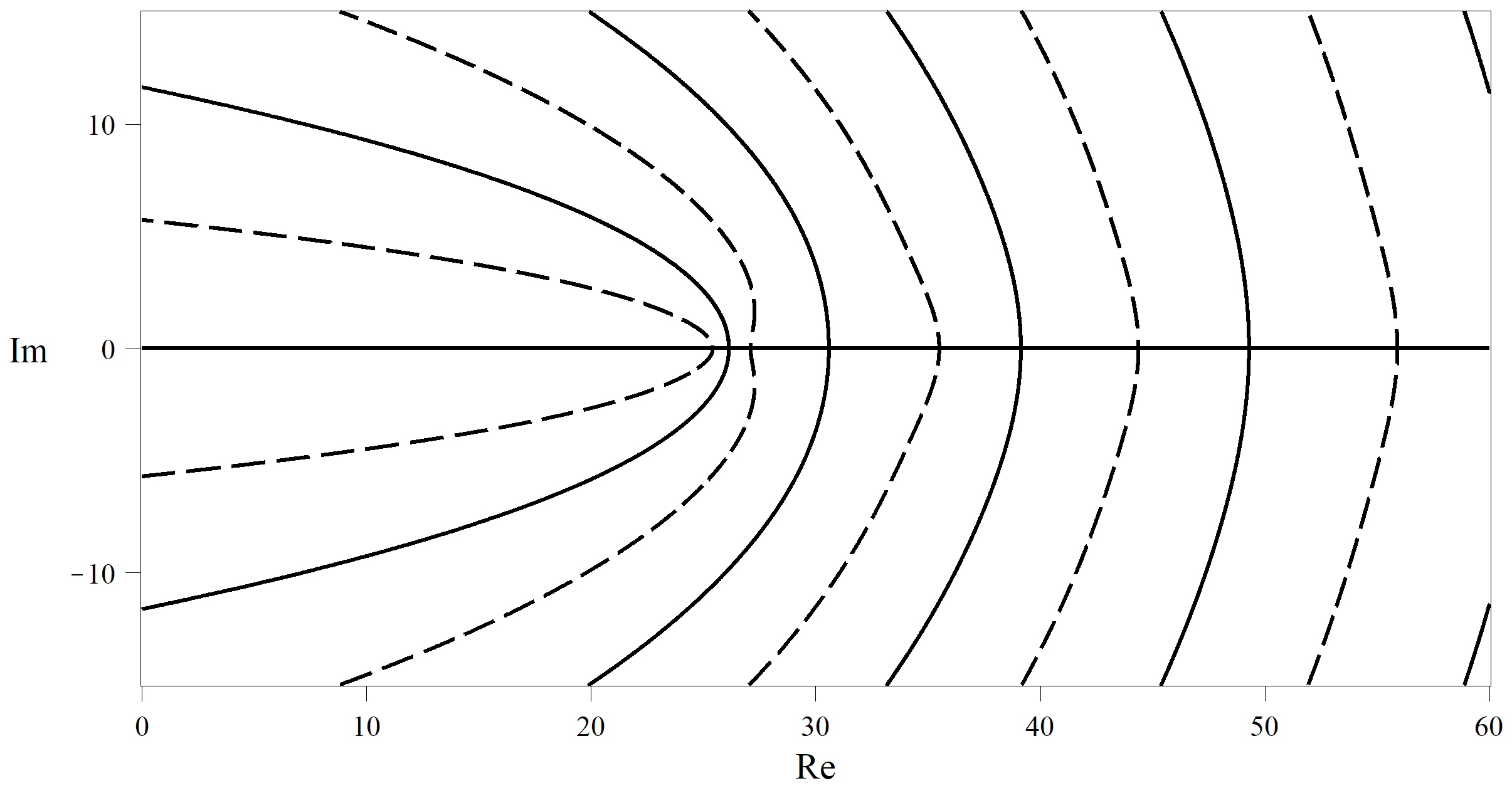}} \\
\subfloat[$\mu=2+0.05\,\I$]{\includegraphics{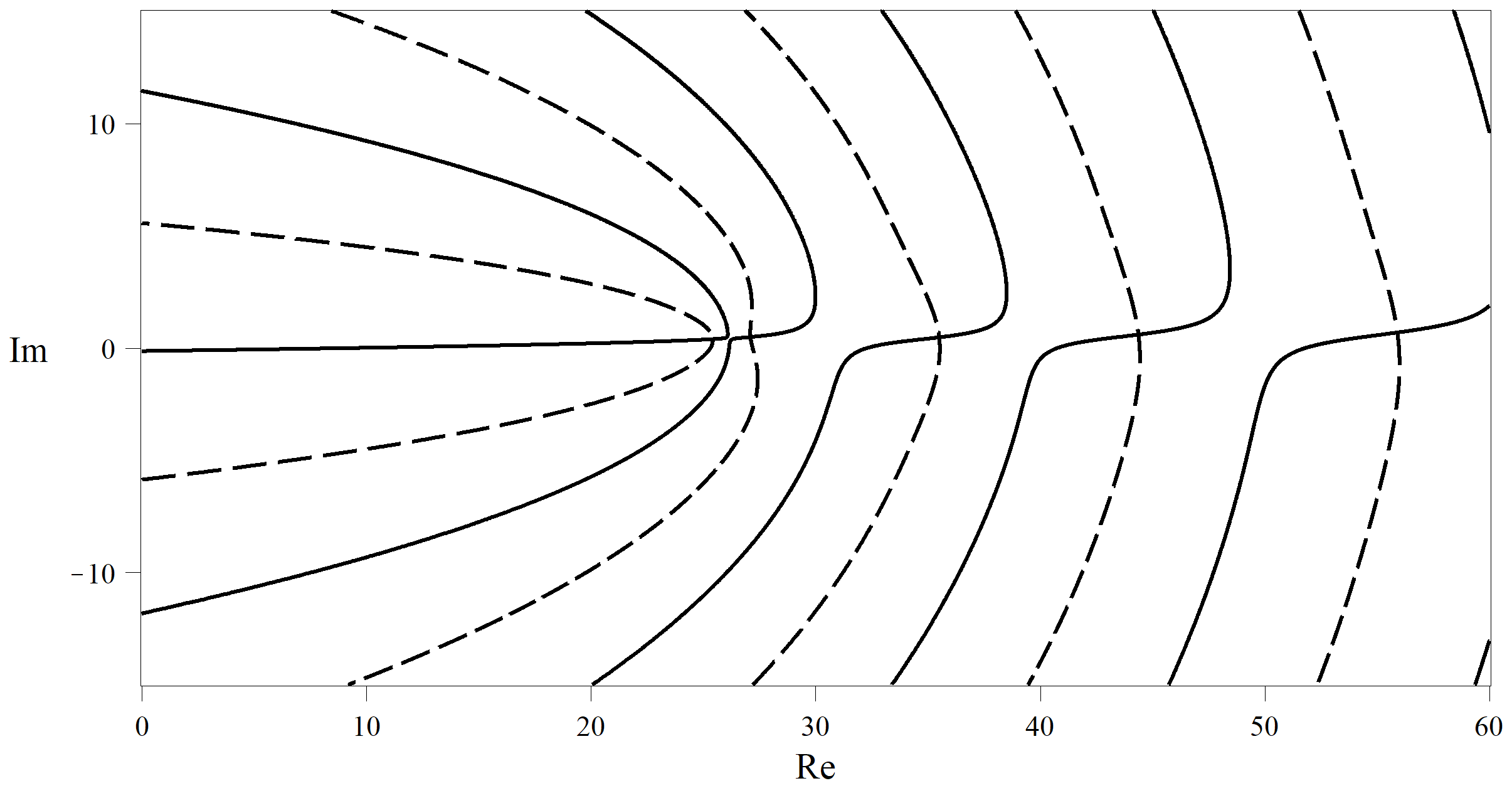}}
\caption{The zeros of the functions $\re\tilde\Theta(\lambda)$ (dashed) and $\im\tilde\Theta(\lambda)$ (solid) for two different parameters $\mu$ and fixed values $\beta=0$, $\gamma^2=-25$ within the range $0\leq\re\lambda\leq 60$, $-15\leq\im\lambda\leq 15$ in the complex $\lambda$-plane.}
\label{fig:Theta-complex}
\end{figure}

\begin{table}[tbhp]\medskip
\begin{tabular}{c|c}
 $\mu = 2$       & $\mu = 2+0.05\,\I$ \\ \hline & \\[-2ex]
 $25.4289571085$ & $25.4290583061+0.3844748370\,\I$ \\
 $27.1098058160$ & $27.1087295464+0.4786514091\,\I$ \\
 $35.5123673338$ & $35.5086680718+0.4658209197\,\I$ \\
 $44.3843905254$ & $44.3817462437+0.5879425852\,\I$ \\
 $55.9101722480$ & $55.9074629810+0.6802438797\,\I$
\end{tabular}
\caption{The eigenvalues of the angular spheroidal wave equation \eqref{ASWE} for $\gamma^2=-25$ and two sample values $\mu$ within the range $0\leq\re\lambda\leq 60$, $-15\leq\im\lambda\leq 15$.}\label{tab:lambda2n}
\end{table}

Finally, let us compare the numerical results for the integer case $\mu=2$ with the values listed in some publications, cf. \cref{tab:Lambda2n}. Unfortunately, there is no commonly accepted standard form for the angular spheroidal wave equation. In the present paper we follow the notation \eqref{ASWE} established by J. Meixner and F. W. Schäfke \cite[Chapter 3]{MS:1954}, which is well suited for the general case of complex parameters; it is also used in \cite{FAW:2003}, \cite{NIST:2010} or \cite{KPS:1976}, for instance. Another frequently encountered notation is that of Flammer \cite{Flammer:1957}:
\begin{equation} \label{FSWE}
\D{}{x}\left((1-x^2)\D{}{x}w(x)\right) + \left(\Lambda - \gamma^2 x^2 - \frac{\mu^2}{1-x^2}\right)w(x) = 0,\quad -1<x<1
\end{equation}
which is applied e.g. in \cite{Falloon:2001}, \cite{Skoro:2015} or \cite{AS:1972}. Moreover, a lot of numerical tables, such as \cite{SL:1964} or \cite{TSWF:1975}, refer to the form \eqref{FSWE}. 
Obviously, the eigenvalues of \eqref{ASWE} and \eqref{FSWE} are simply related by $\lambda = \Lambda-\gamma^2$.

\begin{table}[tbhp]\medskip
\begin{tabular}{c|c|c}
 Flammer &  Zhang \& Jin & Stuckey \& Layton \\
\cite[Tables 131, 132]{Flammer:1957} & \cite[Table 15.15]{ZJ:1996} & \cite[Table 13, $c=5.0$]{SL:1964}  \\ \hline & \\[-2ex]
 $0.42896$ & $0.42895711$ & $0.42895710850$ \\
 $2.10982$ & $2.1098058$  & $21.098058160$  \\
 $10.512$  & $10.512367$  & $10.512367333$  \\
 $19.384$  & $19.384391$  & $19.384390525$  \\
 $30.911$  & $30.910172$  & $30.910172248$
\end{tabular}
\caption{The five lowest spheroidal eigenvalues $\Lambda$ of the angular spheroidal wave equation \eqref{FSWE} in the oblate case $\gamma^2=-25$ for $\mu=2$ taken from different numerical tables. By means of $\lambda=\Lambda-\gamma^2=\Lambda+25$, we obtain the eigenvalues $\lambda$ of \eqref{ASWE} listed in \cref{tab:lambda2n}.}\label{tab:Lambda2n}
\end{table}

It should be noted that the Coulomb spheroidal wave equation can also be written in a slightly different form. By means of the transformation $w(x) = (1-x^2)^{\mu/2}\psi(x)$, \eqref{CSWE} is equivalent to
\begin{equation} \label{ChuStr}
(1-x^2)\psi''(x) - 2(\mu+1)x\,\psi'(x) + \left(t + \beta x + \gamma^2(1-x^2)\right)\psi(x) = 0
\end{equation}
for $-1<x<1$, where $t = \lambda-\mu(\mu+1)$ appears as the eigenvalue parameter. This ODE has a nontrivial bounded solution if and only if $\lambda=t+\mu(\mu+1)$ is an eigenvalue of \eqref{GSWE}. Hence, the eigenvalues of \eqref{GSWE} are exactly the zeros of the function $\Theta(t)$ defined in \cref{thm:MainRes} or \cref{cor:NumCom}. The differential equation \eqref{ChuStr} is a generalization of the angular spheroidal wave equation, written in a notation that goes back to Chu and Stratton \cite[Section 1]{CS:1941}.

For further numerical computations it may be useful to examine the asymptotic behavior of the function $\Theta(t)$ for large $|t|$. Here we will consider only the special case $\beta=\gamma=0$ and $\mu\in\R$.

\begin{Lemma}
If $\beta=\gamma=0$ and $\mu\in\R$, $\mu\geq 0$, then we obtain for $t\in\R$ the asymptotic behavior
\begin{equation*}
\Theta(t) \sim \left\{\begin{array}{ll}
\frac{\Gamma(\mu+1)^2}{\pi}\,t^{-\mu}\cos\big((\sqrt{t}-\mu)\pi\big), & \quad t\to+\infty \\[1ex]
\frac{\Gamma(\mu+1)^2}{2\pi}\,|t|^{-\mu}\E^{\pi\sqrt{|t|}}, & \quad t\to-\infty
\end{array}\right.
\end{equation*}
\end{Lemma}

\begin{proof}
From \cite[Sec. 1.1]{MOS:1966} it follows that $\frac{\Gamma(z+a)}{\Gamma(z+b)}\sim z^{a-b}\left(1+\Osym(\tfrac{1}{z})\right)$ for $z\to\infty$, $|\arg z|<\pi$, and \cref{thm:MainRes}, (d) yields
\begin{equation*}
\Theta(t) 
= \cos\left((\tau-\mu)\pi\right)\frac{\Gamma(\mu+1)^2}{\pi}\frac{\Gamma(\tau+\frac{1}{2}-\mu)}{\Gamma(\tau+\frac{1}{2}+\mu)} \sim \frac{\Gamma(\mu+1)^2}{\pi}\,\tau^{-2\mu}\cos\left((\tau-\mu)\pi\right)
\end{equation*}
For large real numbers $t\to+\infty$ also $\tau=\sqrt{t+(\mu+\tfrac{1}{2})^2}=\sqrt{t}+\Osym(t^{-1/2})$ is real, and hence
\begin{equation*}
\Theta(t)\sim\frac{\Gamma(\mu+1)^2}{\pi}\,t^{-\mu}\cos\big((\sqrt{t}-\mu)\pi\big)\quad\mbox{for real}\quad t\to+\infty
\end{equation*}
Now let us study the asymptotic behavior for $t\to-\infty$. In this case $\tau$ becomes a purely imaginary number with $\im(\tau) > 0$. According to \cite[5.11.9]{NIST:2010}, we have $|\Gamma(x+\I\,y)|\sim\sqrt{2\pi}\,|y|^{x-(1/2)}\E^{-\pi|y|/2}$ for real $x$ and $y\to\pm\infty$. Therefore,
\begin{equation*}
\Theta(t) 
= \frac{\Gamma(\mu+1)^2}{\Gamma(\mu+\frac{1}{2}+\tau)\Gamma(\mu+\frac{1}{2}-\tau)}
\sim \frac{\Gamma(\mu+1)^2}{2\pi|\tau|^{2\mu}\E^{-\pi|\tau|}}
\end{equation*}
and by means of $|\tau|=\sqrt{|t|}+\Osym(|t|^{-1/2})$, we obtain
\begin{equation*}
\Theta(t) \sim \frac{\Gamma(\mu+1)^2}{2\pi}\,|t|^{-\mu}\E^{\pi\sqrt{|t|}}\quad\mbox{for real}\quad t\to-\infty
\end{equation*}
\end{proof}

\section{Proof of the main theorem}
\label{sec:Proof}

In this section we study the CSWE \eqref{CSWE}. We assume that $\mu,\beta,\gamma\in\C$ are fixed numbers, whereas $\lambda\in\C$ is considered to be the eigenvalue parameter. Initially, we will associate a first order $2\times 2$ system to the second order ODE \eqref{CSWE}. To this end, we introduce the function
\begin{equation} \label{DefV}
v(x) := 2w'(x) + \frac{2\mu x}{1-x^2}\,w(x)
\end{equation}
If $w(x)$ is a solution of \eqref{CSWE}, then
\begin{align*}
& \D{}{x}\left((1-x^2)v(x)\right) 
  = 2\,\D{}{x}\left((1-x^2)\D{}{x}w(x)\right) + 2\mu w(x) + 2\mu x w'(x) \\
& = -2\left(\lambda + \beta x + \gamma^2(1-x^2) - \frac{\mu^2}{1-x^2}\right)w(x) + 2\mu w(x)
    + \mu x\left(v(x)-\frac{2\mu x}{1-x^2}\,w(x)\right)  \\
& = \mu x v(x) - 2\left(\lambda - \mu(\mu+1) + \beta x + \gamma^2(1-x^2)\right)w(x)
\end{align*}
Solving this equation for $v'(x)$ and \eqref{DefV} for $w'(x)$, we get the differential system
\begin{alignat*}{3}
v'(x) & {}={} & -\frac{(\mu+2)x}{x^2-1}\,v(x) & +\frac{2(t + \beta x) - 2\gamma^2(x^2-1)}{x^2-1}\,w(x) \\
w'(x) & {}={} & \frac{1}{2}\,v(x) & +\frac{\mu x}{x^2-1}\,w(x)
\end{alignat*}
where $t := \lambda-\mu(\mu+1)$. Moreover, by means of the transformation $x = 2z-1$, the vector function
\begin{equation*}
y(z) = \begin{pmatrix} v(2z-1) \\[1ex] w(2z-1) \end{pmatrix}
\end{equation*}
is a solution of the $2\times 2$ system
\begin{equation} \label{RegSing}
y'(z) = \left(\frac{1}{z}\begin{pmatrix} -\frac{\mu}{2}-1 & \beta-t \\[1ex] 0 & \frac{\mu}{2} \end{pmatrix} + \frac{1}{z-1}\begin{pmatrix} -\frac{\mu}{2}-1 & \beta+t \\[1ex] 0 & \frac{\mu}{2} \end{pmatrix} + \begin{pmatrix} 0 & -4\gamma^2 \\[1ex] 1 & 0 \end{pmatrix}\right)y(z)
\end{equation}
This is a meromorphic differential system in the complex plane with regular singular points at $z=0$ and $z=1$, and an irregular singularity at infinity. Conversely, if $y(z)$ is a solution of \eqref{RegSing}, then its first component satisfies \eqref{DefV} for $x=2z-1$, and we can derive the CSWE \eqref{CSWE} for its second component $w(x)$. The next step is to find appropriate fundamental matrices to \eqref{RegSing}, and for this purpose we use

\begin{Lemma} \label{lem:StructFM}
Suppose that $\mu\in\C\setminus\{-1,-2,-3,\ldots\}$, and let $R:\mathfrak{D}\longrightarrow\MzC$ be a holomorphic matrix function on the open disk $\mathfrak{D} := \{z\in\C:|z|<\varepsilon\}$ for some $\varepsilon>0$. Then the differential system
\begin{equation} \label{RegSingA}
y'(z) = \left(\frac{1}{z}\begin{pmatrix} -\frac{\mu}{2}-1 & \sigma \\[1ex] 0 & \frac{\mu}{2} \end{pmatrix} + R(z)\right)y(z),\quad 
R(0) = \begin{pmatrix} \ast & \ast \\[1ex] \rho & \ast \end{pmatrix}
\end{equation}
with arbitrary $\sigma,\rho\in\C$ has a fundamental matrix of the form
\begin{equation} \label{FundMatA}
Y(z) = \begin{pmatrix} \frac{1}{z} & 0 \\[1ex] 0 & 1 \end{pmatrix}H(z)
\begin{pmatrix} z^{-\mu/2} & 0 \\[1ex] 0 & z^{\mu/2} \end{pmatrix}
\begin{pmatrix} 1 & 0 \\[1ex] q\log z & 1 \end{pmatrix},\quad
H(0) = \begin{pmatrix} 1 & 0 \\[1ex] p & 1 \end{pmatrix}
\end{equation}
where $H:\mathfrak{D}\longrightarrow\MzC$ is a holomorphic matrix function and $p$, $q$ are some complex numbers. If $\mu$ is not an integer, then $q=0$, and if $\mu\neq 0$, then $p=-\frac{\rho}{\mu}$; in case of $\mu=0$ we get $p=0$, $q=\rho$.
\end{Lemma}

\begin{proof}
First we consider the case $\mu\neq 0$. By means of the transformation
\begin{equation} \label{ShearA}
y(z) = 
z^{-\mu/2}\begin{pmatrix} \ms\frac{1}{z} & 0 \\[1ex] -\frac{\rho}{\mu} & 1 \end{pmatrix}y_0(z)
\end{equation}
the differential system \eqref{RegSingA} is equivalent to
\begin{equation} \label{RegSingD}
y_0'(z)=\tfrac{1}{z}\Phi_0(z)y_0(z),\quad
\Phi_0(z) = \sum_{k=0}^\infty z^k\Phi_{0,k},\quad
\Phi_{0,0} = D := \begin{pmatrix} 0 & 0 \\[1ex] 0 & \mu \end{pmatrix}
\end{equation}
where the coeffient matrix $\Phi_0(z)$ is a holomorphic matrix function on the disc $\mathfrak{D}$.

If $\mu\not\in\Z$, then the eigenvalues of the diagonal matrix $D$ do not differ by an integer, and the system \eqref{RegSingD} has a fundamental matrix of the form $Y_0(z) = G(z)z^D$, where $G:\mathfrak{D}\longrightarrow\MzC$ is a holomorphic matrix function with $G(0)=I$ (cf. \cite[Theorem 5.5]{Wasow:1976}). Thus, with regard to \eqref{ShearA},
\begin{equation*}
Y(z) = 
z^{-\mu/2}\begin{pmatrix} \ms\frac{1}{z} & 0 \\[1ex] -\frac{\rho}{\mu} & 1 \end{pmatrix}G(z)
\begin{pmatrix} 1 & 0 \\[1ex] 0 & z^{\mu} \end{pmatrix} = 
\begin{pmatrix} \frac{1}{z} & 0 \\[1ex] 0 & 1 \end{pmatrix}H(z)
\begin{pmatrix} z^{-\mu/2} & 0 \\[1ex] 0 & z^{\mu/2} \end{pmatrix}
\end{equation*}
is a fundamental matrix of \eqref{RegSingA} having the form \eqref{FundMatA} with $q=0$, where
\begin{equation} \label{TrafoH}
H(z) := \begin{pmatrix} \ms 1 & 0 \\[1ex] -\frac{\rho}{\mu} & 1 \end{pmatrix}G(z),\quad 
H(0) = \begin{pmatrix} \ms 1 & 0 \\[1ex] -\frac{\rho}{\mu} & 1 \end{pmatrix}
\end{equation}
Now, we consider the case that $\mu := m$ is a positive integer. If we recursively apply the transformations
\begin{equation} \label{ShearB}
y_{n-1}(z) = \begin{pmatrix} 1 & 0 \\[1ex] \frac{\phi_{n-1}}{n-m}z & z \end{pmatrix}y_n(z)
\end{equation}
for $n=1,\ldots,m-1$, then the vector function $y_n(z)$ is a solution of a differential system
\begin{equation*}
y_n'(z) = \tfrac{1}{z}\Phi_n(z)y_n(z),\quad 
\Phi_n(z) = \sum_{k=0}^\infty z^k\Phi_{n,k},\quad
\Phi_{n,0} = \begin{pmatrix} 0 & 0 \\[1ex] 0 & m-n \end{pmatrix}
\end{equation*}
provided that at each step $\phi_{n-1}$ is taken to be the $(2,1)$-entry of the matrix $\Phi_{n-1,1}$. The coefficient matrix $\Phi_n(z)$ is holomorphic on $\mathfrak{D}$. Moreover, $\Phi_{n,k}$ for $k=1,\ldots,n$ are lower triangular matrices, i.e., their $(1,2)$-entry is zero. Finally, if we apply the shearing transformation
\begin{equation} \label{ShearC}
y_{m-1}(z) = \begin{pmatrix} 1 & 0 \\[1ex] 0 & z \end{pmatrix}y_m(z)
\end{equation}
then we obtain the regular singular $2\times 2$ system
\begin{equation} \label{Jordan}
y_m'(z) = \tfrac{1}{z}\Phi_m(z)y_m(z),\quad
\Phi_m(z) = \sum_{k=0}^\infty z^k\Phi_{m,k},\quad
\Phi_{m,0} = Q := \begin{pmatrix} 0 & 0 \\[1ex] q & 0 \end{pmatrix}
\end{equation}
Here, $q$ is just the $(2,1)$-component of $\Phi_{m-1,1}$. In addition, $\Phi_{m,k}$ are lower triangular matrices for $k=0,\ldots,m$. According to \cite[Theorem 5.5]{Wasow:1976}, the system \eqref{Jordan} has a fundamental matrix of the form
\begin{equation} \label{BlockD}
Y_m(z) = F(z)z^Q = F(z)\begin{pmatrix} 1 & 0 \\[1ex] q\log z & 1 \end{pmatrix}
\end{equation}
where $F(z) = \sum_{k=0}^\infty z^k F_k$ is a holomorphic matrix function on $\mathfrak{D}$ satisfying $F_0=I$. Note that $F(z)$ is a solution of the matrix differential equation
\begin{equation*}
z F'(z) = \Phi_m(z)F(z)-F(z)Q,\quad z\in\mathfrak{D}
\end{equation*}
and the coefficients $F_k\in\MzC$ for $k>0$ are uniquely determined by the recurrence relation
\begin{equation*}
F_kQ - (Q-kI)F_k  = \sum_{n=0}^{k-1}\Phi_{m,k-n}F_n
\end{equation*}
Since $Q$ and $\Phi_{m,k}$ for $k=0,\ldots,m$ are lower triangular matrices, it is easy to verify that $F_k$ are also lower triangular for $k = 0,\ldots,m$. Now, by combining the transformations \eqref{ShearB} and \eqref{ShearC} with \eqref{BlockD}, it follows that the differential system \eqref{RegSingD} has a fundamental matrix of the form
\begin{equation*}
Y_0(z) = 
\begin{pmatrix} 1 & 0 \\[1ex] z\psi(z) & z^m\end{pmatrix} F(z)
\begin{pmatrix} 1 & 0 \\[1ex] q\log z & 1 \end{pmatrix}
\end{equation*}
where $\psi(z)$ is a polynomial in $z$ of degree $m-2$, and $\psi\equiv 0$ in case of $m=1$. Moreover, $F(z)$ can be written the form
\begin{equation*}
F(z) 
= \begin{pmatrix} 1+zf_{11}(z) & z^{m+1}f_{12}(z) \\[1ex] zf_{21}(z) & 1+zf_{22}(z) \end{pmatrix}
\end{equation*}
with some holomorphic functions $f_{ij}(z)$ on $\mathfrak{D}$. If we define
\begin{align*}
G(z) & := 
\begin{pmatrix} 1 & 0 \\[1ex] z\psi(z) & z^m \end{pmatrix} F(z)
\begin{pmatrix} 1 & 0 \\[1ex] 0 & z^{-m} \end{pmatrix} \\ & =
\begin{pmatrix} 1+zf_{11}(z) & zf_{12}(z) \\[1ex] z\psi(z)(1+zf_{11}(z))+z^{m+1}f_{21}(z) & 1+zf_{22}(z)+z^2\psi(z)f_{12}(z) \end{pmatrix}
\end{align*}
then $G:\mathfrak{D}\longrightarrow\MzC$ is holomorphic with $G(0)=I$, and the fundamental matrix of \eqref{RegSingD} becomes
\begin{equation*}
Y_0(z) = G(z)\begin{pmatrix} 1 & 0 \\[1ex] 0 & z^m \end{pmatrix}
\begin{pmatrix} 1 & 0 \\[1ex] q\log z & 1 \end{pmatrix}
\end{equation*}
Applying \eqref{ShearA} and defining $H(z)$ for $\mu=m$ as in \eqref{TrafoH} yields a fundamental matrix of \eqref{RegSing}, which takes the form
\begin{equation*}
Y(z) = 
\begin{pmatrix} \frac{1}{z} & 0 \\[1ex] 0 & 1 \end{pmatrix}H(z)
\begin{pmatrix} z^{-\mu/2} & 0 \\[1ex] 0 & z^{\mu/2} \end{pmatrix}
\begin{pmatrix} 1 & 0 \\[1ex] q\log z & 1 \end{pmatrix}
\end{equation*}
It remains to study the case $\mu=0$. By virtue of the transformation
\begin{equation*}
y(z) = \begin{pmatrix} \frac{1}{z} & 0 \\[1ex] 0 & 1 \end{pmatrix}y_0(z)
\end{equation*}
the system \eqref{RegSingA} with $\mu=0$ is equivalent to
\begin{equation*}
y_0'(z) = \left(\tfrac{1}{z}\,Q 
+ S(z)\right)y_0(z),\quad Q := \begin{pmatrix} 0 & 0 \\[1ex] \rho & 0 \end{pmatrix}
\end{equation*}
with some holomorphic matrix function $S(z)$ on $\mathfrak{D}$. It has a fundamental matrix $Y_0(z) = H(z)z^Q$, where $H:\mathfrak{D}\longrightarrow\MzC$ is holomorphic and $H(0)=I$. Thus,
\begin{equation*}
Y(z) := \begin{pmatrix} \frac{1}{z} & 0 \\[1ex] 0 & 1 \end{pmatrix}H_0(z)z^Q
= \begin{pmatrix} \frac{1}{z} & 0 \\[1ex] 0 & 1 \end{pmatrix}H_0(z)
  \begin{pmatrix} 1 & 0 \\[1ex] \rho\log z & 1 \end{pmatrix}
\end{equation*}
is a fundamental matrix of \eqref{RegSing} for $\mu=0$, which has the form \eqref{FundMatA} with $p=0$ and $q=\rho$.
\end{proof}

\Cref{lem:StructFM} provides the structure of the fundamental matrices for the differential system \eqref{RegSingA}. If, in addition, the coefficient matrix $R$ depends holomorphically on some parameter, then we obtain the following enhancement:

\begin{Lemma} \label{lem:EntireFM}
Let $\mu\in\C$ be a fixed number satisfying $\re\mu>-1$. Suppose that $\sigma=\sigma(t)$ depends holomorphically on some parameter $t\in\C$. Moreover, assume that $R=R(z,t)$ depends on $z\in\mathfrak{D}$ and $t\in\C$, such that $R:\mathfrak{D}\times\C\longrightarrow\MzC$ is a holomorphic matrix function. Then the differential system \eqref{RegSingA} has a fundamental matrix of the form \eqref{FundMatA}, where $H:\mathfrak{D}\times\C\longrightarrow\MzC$ is a holomorphic matrix function, and also $p=p(t)$, $q=q(t)$ depend holomorphically on $t\in\C$.
\end{Lemma}

\begin{proof}
The differential system \eqref{RegSingA} has the form
\begin{equation*}
\frac{\partial y}{\partial z}(z,t) = \frac{1}{z}\Psi(z,t)y(z,t),\quad 
(z,t)\in\left(\mathfrak{D}\setminus\{0\}\right)\times\C
\end{equation*}
where the coefficient matrix
\begin{equation*}
\Psi(z,t) = 
\begin{pmatrix} -\frac{\mu}{2}-1 & \sigma(t) \\[1ex] 0 & \frac{\mu}{2} \end{pmatrix} + zR(z,t)
\end{equation*}
is a holomorphic function on $\mathfrak{D}\times\C$. The eigenvalues $-\frac{\mu}{2}-1$ and $\frac{\mu}{2}$ of $\Psi(0,t)$ are distinct and independent of $t\in\C$. In particular, $\re\mu>-1$ implies $\re(-\frac{\mu}{2}-1)<\re\frac{\mu}{2}$, and we have
\begin{equation*}
\tilde G(t)^{-1}\Psi(0,t)\tilde G(t) = D := \begin{pmatrix} -\frac{\mu}{2}-1 & 0 \\[1ex] 0 & \frac{\mu}{2} \end{pmatrix}
\quad\mbox{with}\quad
\tilde G(t) := \begin{pmatrix} 1 & \frac{\sigma(t)}{\mu+1} \\[1ex] 0 & 1 \end{pmatrix}
\end{equation*}
From \cite[Lemma 6]{BSW:2004} it follows that \eqref{RegSingA} has a fundamental matrix of the form
\begin{equation} \label{FundMatB}
\tilde G(t)\tilde H(z,t)\begin{pmatrix} z^{-\mu/2-1} & 0 \\[1ex] 0 & z^{\mu/2} \end{pmatrix}
\begin{pmatrix} 1 & 0 \\[1ex] \tilde q(t)\log z & 1 \end{pmatrix}
\end{equation}
where $\tilde H:\mathfrak{D}\times\C\longrightarrow\MzC$, $\tilde q:\mathfrak{D}\longrightarrow\C$ are holomorphic functions and $\tilde H(0,t)=I$ for all $t\in\C$. Comparing \eqref{FundMatB} to the fundamental matrix
\begin{equation*}
\begin{pmatrix} \frac{1}{z} & 0 \\[1ex] 0 & 1 \end{pmatrix}H(z,t)
\begin{pmatrix} z^{-\mu/2} & 0 \\[1ex] 0 & z^{\mu/2} \end{pmatrix}
\begin{pmatrix} 1 & 0 \\[1ex] q(t)\log z & 1 \end{pmatrix}
\end{equation*}
given by \eqref{FundMatA}, it follows that $q(t)=\tilde q(t)$, and 
\begin{equation*}
H(z,t) = \begin{pmatrix} z & 0 \\[1ex] 0 & 1 \end{pmatrix}\tilde G(t)\tilde H(z,t)\begin{pmatrix}  \frac{1}{z} & 0 \\[1ex] 0 & 1 \end{pmatrix}
\end{equation*}
is a holomorphic matrix function with respect to $(z,t)\in\mathfrak{D}\times\C$.
\end{proof}

We can apply \cref{lem:StructFM} straightforwardly to the differential system \eqref{RegSing}, where $\sigma = \beta - t$, $\rho = 1$, and
\begin{equation*}
R(z) := \frac{1}{z-1}\begin{pmatrix} -\frac{\mu}{2}-1 & \beta+t \\[1ex] 0 & \frac{\mu}{2} \end{pmatrix} + \begin{pmatrix} 0 & -4\gamma^2 \\[1ex] 1 & 0 \end{pmatrix}\quad\mbox{with}\quad
R(0) = \begin{pmatrix} \ast & \ast \\[1ex] 1 & \ast \end{pmatrix}
\end{equation*}
is a holomorphic matrix function on the unit disk $\mathfrak{D}_0:=\{z\in\C:|z|<1\}$. If $\mu\in\C\setminus\{-1,-2,-3,\ldots\}$, then it has a fundamental matrix of the form
\begin{equation} \label{FundMat0}
Y_0(z) = \begin{pmatrix} \frac{1}{z} & 0 \\[1ex] 0 & 1 \end{pmatrix}H_0(z)
\begin{pmatrix} z^{-\mu/2} & 0 \\[1ex] 0 & z^{\mu/2} \end{pmatrix}
\begin{pmatrix} 1 & 0 \\[1ex] q_0\log z & 1 \end{pmatrix},\quad
H_0(0) = \begin{pmatrix} 1 & 0 \\[1ex] p_0 & 1 \end{pmatrix}
\end{equation}
where $H_0:\mathfrak{D}_0\longrightarrow\MzC$ is a holomorphic matrix function. Moreover, we get $q_0=0$ for the case $\mu\not\in\Z$, and $p_0=-\frac{1}{\mu}$ for the case $\mu\neq 0$; if $\mu=0$, then $p_0=0$ and $q_0=1$.

Similarly, \cref{lem:StructFM} yields a fundamental matrix in a neighborhood of $z=1$: By applying the transformation $\tilde y(z) := y(1-z)$, \eqref{RegSing} is equivalent to
\begin{equation*}
\tilde y'(z) = \left(\frac{1}{z}\begin{pmatrix} -\frac{\mu}{2}-1 & \beta+t \\[1ex] 0 & -\frac{\mu}{2} \end{pmatrix} + \frac{1}{z-1}\begin{pmatrix} -\frac{\mu}{2}-1 & \beta-t \\[1ex] 0 & \frac{\mu}{2} \end{pmatrix} + \begin{pmatrix} \ms 0 & 4\gamma^2 \\[1ex] -1 & 0 \end{pmatrix}\right)\tilde y(z)
\end{equation*}
This system has the form \eqref{RegSingA} with $\sigma := \beta+t$ and $\rho = -1$. If $\mu\in\C\setminus\{-1,-2,-3,\ldots\}$, then \cref{lem:StructFM} provides a fundamental matrix for \eqref{RegSing} of the form
\begin{equation} \label{FundMat1}
Y_1(z) = \begin{pmatrix} \frac{1}{1-z} & 0 \\[1ex] 0 & 1 \end{pmatrix}H_1(z)
\begin{pmatrix} (1-z)^{-\mu/2} & 0 \\[1ex] 0 & (1-z)^{\mu/2} \end{pmatrix}
\begin{pmatrix} 1 & 0 \\[1ex] q_1\log(1-z) & 1 \end{pmatrix}
\end{equation}
where $H_1:\mathfrak{D}_1\longrightarrow\MzC$ is a holomorphic matrix function on the unit disk $\mathfrak{D}_1:=\{z\in\C:|z-1|<1\}$ centered at $z=1$. In addition,
\begin{equation*}
H_1(1) = \begin{pmatrix} 1 & 0 \\[1ex] p_1 & 1 \end{pmatrix}
\end{equation*}
where $q_1=0$ for the case $\mu\not\in\Z$ and $p_1=\frac{1}{\mu}$ for the case $\mu\neq 0$; if $\mu=0$, then $p_1=0$ and $q_1=-1$.

In a next step, we will extract the solutions of \eqref{RegSing} for which the second component is bounded at both singular points $z=0$ and $z=1$.

\begin{Lemma} \label{lem:RegSol0}
If $\re\mu>0$ or $\mu=0$, then the differential system \eqref{RegSing} has a Floquet solution
\begin{equation} \label{FloqSol0}
y_0(z) = z^{\mu/2}h_0(z),\quad h_0(z) = \sum_{k=0}^\infty z^k a_k,\quad
a_0 := \begin{pmatrix} \frac{\beta-t}{\mu+1} \\[1ex] 1 \end{pmatrix}
\end{equation}
where $h_0:\mathfrak{D}_0\longrightarrow\Cz$ is a holomorphic vector function. If, in addition, $Y_1(z)$ denotes the fundamental matrix \eqref{FundMat1}, then
\begin{equation} \label{FundCon1}
y_0(z) = Y_1(z)c\quad\mbox{with some vector}\quad c = \begin{pmatrix} c_1 \\[1ex] c_2 \end{pmatrix}\in\Cz
\end{equation}
where the connection coefficients $c_1=c_1(t)$, $c_2=c_2(t)$ depend holomorphically on the parameter $t\in\C$. Finally, $\lambda=t+\mu(\mu+1)$ is an eigenvalue of the CSWE \eqref{CSWE} if and only if $c_1(t)=0$; in this case, $y_0(z)$ is a constant multiple of 
\begin{equation} \label{FloqSol1}
Y_1(z)e_2 = (1-z)^{\mu/2}\sum_{k=0}^\infty (1-z)^k b_k,\quad
b_0 := \begin{pmatrix} \frac{\beta+t}{\mu+1} \\[1ex] 1 \end{pmatrix}
\end{equation}
\end{Lemma}

\begin{proof}
The system \eqref{RegSing} has the form $y'(z) = \big(\frac{1}{z}\,A + R(z)\big)y(z)$ with 
\begin{equation*}
A := \begin{pmatrix} -\frac{\mu}{2}-1 & \beta-t \\[1ex] 0 & \frac{\mu}{2} \end{pmatrix}
\quad\mbox{and}\quad
R(z) := \frac{1}{z-1}\begin{pmatrix} -\frac{\mu}{2}-1 & \beta+t \\[1ex] 0 & \frac{\mu}{2} \end{pmatrix} + \begin{pmatrix} 0 & -4\gamma^2 \\[1ex] 1 & 0 \end{pmatrix}
\end{equation*}
Here, $R(z)$ is a holomorphic matrix function on $\mathfrak{D}_0$, and $\frac{\mu}{2}$ is an eigenvalue of $A$ but $\frac{\mu}{2}+k$ is not an eigenvalue for any $k\in\N$. Moreover, $a_0$ is an eigenvector of $A$ corresponding to $\frac{\mu}{2}$. According to \cite[§24.VI.(d)]{Walter:1998}, 
the system \eqref{RegSing} has a solution as indicated in \eqref{FloqSol0}. Note that this solution coincides with the fundamental solution $Y_0(z)e_2$, where $Y_0(z)$ denotes the fundamental matrix \eqref{FundMat0}. Another  fundamental solution is given by
\begin{equation*}
Y_0(z)e_1 = 
\begin{pmatrix} \frac{1}{z} & 0 \\[1ex] 0 & 1 \end{pmatrix}H_0(z)
\begin{pmatrix} z^{-\mu/2} \\[1ex] q_0 z^{\mu/2}\log z \end{pmatrix},\quad
H_0(z) = \begin{pmatrix} 1+o(1) & o(1) \\[1ex] p_0+o(1) & 1+o(1) \end{pmatrix}
\end{equation*}
The second component of this vector function is not bounded for $z\to 0$ since
\begin{equation*}
Y_0(z)e_1 = \begin{pmatrix} \ast \\[1ex] 
z^{-\mu/2}(p_0+o(1)) + q_0(1+o(1))z^{\mu/2}\log z \end{pmatrix}
\end{equation*}
where $p_0=-\frac{1}{\mu}\neq 0$ in case of $\re\mu>0$ and $p_0=0$, $q_0=1$ for $\mu=0$. By definition, $\lambda\in\C$ is an eigenvalue of \eqref{CSWE} if and only if the system \eqref{RegSing} has a nontrivial solution $y:(0,1)\longrightarrow\MzC$ where its second component is bounded on $(0,1)$. In particular, such a solution $y(z)$ must be a constant multiple of $y_0(z)$. Moreover, there is another set of fundamental solutions, namely $y_1(z):=Y_1(z)e_1$ and $y_2(z):=Y_1(z)e_2$, and hence the solution $y_0(z)$ can be written as a linear combination $y_0(z)=c_1 y_1(z)+c_2 y_2(z)$ with some connection coefficients $c_1,c_2\in\C$. Note that $Y_0(z)e_2=y_0(z)=Y_1(z)c$ holds for any $z\in\mathfrak{D}_0\cap\mathfrak{D}_1$, and in particular we obtain $c = Y_1(z)^{-1}Y_0(z)e_2$ for $z=\frac{1}{2}$. Since $Y_0(\frac{1}{2})$ and $Y_1(\frac{1}{2})$ depend holomorphically on $t\in\C$ by means of \cref{lem:EntireFM}, it follows that $c=c(t)$ is an entire vector function, i.e., also its entries $c_1(t)$ and $c_2(t)$ depend holomorphically on $t\in\C$. Furthermore, the holomorphic part in $Y_1(z)$ asymptotically behaves like
\begin{equation*}
H_1(z) =  \begin{pmatrix} 1+o(1) & o(1) \\[1ex] p_1+o(1) & 1+o(1) \end{pmatrix}\quad\mbox{as}\quad z\to 1
\end{equation*}
Multiplying \eqref{FundMat1} from the left by the unit vector $e_1$ gives
\begin{equation} \label{FundSol1}
y_1(z) = \begin{pmatrix} \frac{1}{1-z} & 0 \\[1ex] 0 & 1 \end{pmatrix}\begin{pmatrix} 1+o(1) & o(1) \\[1ex] p_1+o(1) & 1+o(1) \end{pmatrix}\begin{pmatrix} (1-z)^{-\mu/2} \\[1ex] q_1 (1-z)^{\mu/2}\log(1-z) \end{pmatrix}
\end{equation}
which yields the asymptotic behaviour
\begin{equation} \label{Asymp1}
y_1(z) = \begin{pmatrix} \ast \\[1ex] 
(1-z)^{-\mu/2}(p_1+o(1)) + q_1(1+o(1))(1-z)^{\mu/2}\log(1-z) \end{pmatrix}
\end{equation}
as $z\to 1$, where $p_1=\frac{1}{\mu}$ if $\re\mu>0$ and $p_1=0$, $q_1=-1$ in case of $\mu=0$. Hence, the second component of $y_1(z)$ is not bounded near $z=1$. On the other hand, the second component of
\begin{equation} \label{FundSol2}
y_2(z) = 
\begin{pmatrix} \frac{1}{1-z} & 0 \\[1ex] 0 & 1 \end{pmatrix}H_1(z)
\begin{pmatrix} 0 \\[1ex] (1-z)^{\mu/2} \end{pmatrix}
= \begin{pmatrix} \ast \\[1ex] (1-z)^{\mu/2}(1+o(1)) \end{pmatrix}
\end{equation}
is bounded as $z\to 1$. Therefore, $\lambda\in\C$ is an eigenvalue of \eqref{CSWE} if and only if the nontrivial solution $y(z)$ is a constant multiple of $y_0(z)$ and also of $y_2(z)$, which means that $c_1=0$. Further, we can write the system \eqref{RegSing} in the form $y'(z) = \big(\frac{1}{z-1}\,B + S(z)\big)y(z)$ with 
\begin{equation*}
B := \begin{pmatrix} -\frac{\mu}{2}-1 & \beta+t \\[1ex] 0 & \frac{\mu}{2} \end{pmatrix}
\quad\mbox{and}\quad
S(z) := \frac{1}{z}\begin{pmatrix} -\frac{\mu}{2}-1 & \beta-t \\[1ex] 0 & \frac{\mu}{2} \end{pmatrix} + \begin{pmatrix} 0 & -4\gamma^2 \\[1ex] 1 & 0 \end{pmatrix}
\end{equation*}
Here, $\mu$ is an eigenvalue of $B$, while $\frac{\mu}{2}+k$ is not an eigenvalue for any $k\in\N$. Since $S(z)$ is holomorphic on $\mathfrak{D}_1$, there exists a Floquet solution having the form $(1-z)^{\mu/2}\sum_{k=0}^\infty (1-z)^k b_k$ according to \cite[§24.VI.(d)]{Walter:1998}, where $b_0$ must be an eigenvector of $B$ for the eigenvalue $\frac{\mu}{2}$. Comparing this Floquet solution to \eqref{FundSol2} yields \eqref{FloqSol1}, which completes the proof of \cref{lem:RegSol0}.
\end{proof}

Our next aim is to simplify the system \eqref{RegSing} so that the computation of the Floquet solutions and their connection coefficients becomes as simple as possible. For this purpose, we apply the transformation
\begin{equation} \label{TrafoE}
\eta(z) = z^{-\mu/2}(1-z)^{\mu/2}y(z)
\end{equation}
which turns \eqref{RegSing} into the differential system
\begin{equation} \label{RegSing0}
\eta'(z) = \left(\frac{1}{z}\begin{pmatrix} -\mu-1 & \beta-t \\[1ex] 0 & 0 \end{pmatrix} + \frac{1}{z-1}\begin{pmatrix} -1 & \beta+t \\[1ex] \ms 0 & \mu \end{pmatrix} + \begin{pmatrix} 0 & -4\gamma^2 \\[1ex] 1 & 0 \end{pmatrix}\right)\eta(z)
\end{equation}
It can be written in the form $\eta'(z) = \big(\frac{1}{z}A_0 + \frac{1}{z-1}A_1 + C\big)\eta(z)$ with the coefficient matrices
\begin{equation} \label{ResMat}
A_0 := \begin{pmatrix} -\mu-1 & \beta-t \\[1ex] 0 & 0 \end{pmatrix},\quad
A_1 := \begin{pmatrix} -1 & \beta+t \\[1ex] \ms 0 & \mu \end{pmatrix},\quad
  C := \begin{pmatrix} 0 & -4\gamma^2 \\[1ex] 1 & 0 \end{pmatrix}
\end{equation}
According to \cref{lem:RegSol0} and \eqref{TrafoE}, it has a holomorphic solution
\begin{equation} \label{HolSol}
\eta_0(z) := (1-z)^{\mu/2}h_0(z) = \sum_{k=0}^\infty z^k d_k, \quad d_0 = a_0 = \begin{pmatrix} \frac{\beta-t}{\mu+1} \\[1ex] 1 \end{pmatrix}
\end{equation}
on $\mathfrak{D}_0$ and a fundamental matrix $Y(z):=z^{-\mu/2}(1-z)^{\mu/2}Y_1(z)$, which takes the form
\begin{equation} \label{FundMat}
Y(z) = \begin{pmatrix} \frac{1}{1-z} & 0 \\[1ex] 0 & 1 \end{pmatrix}H(z)
\begin{pmatrix} 1 & 0 \\[1ex] 0 & (1-z)^{\mu} \end{pmatrix}
\begin{pmatrix} 1 & 0 \\[1ex] q_1\log(1-z) & 1 \end{pmatrix}
\end{equation}
where $Y_1(z)$ is given by \eqref{FundMat1} and $H(z) := z^{-\mu/2}H_1(z)$ is holomorphic on $\mathfrak{D}_1$. Moreover, \eqref{FundCon1} implies
\begin{equation*}
\eta_0(z) = Y(z)c = c_1 Y(z)e_1 + c_2 Y(z)e_2,\quad\mbox{where}\quad
c = \begin{pmatrix} c_1 \\[1ex] c_2 \end{pmatrix}
\end{equation*}
is a vector which contains the connection coefficients. According to \cref{lem:RegSol0}, $\lambda\in\C$ is an eigenvalue of \eqref{CSWE} if and only if $c_1=0$. For this case, $\eta_0(z)$ is a constant multiple of the vector function $Y(z)e_2=z^{-\mu/2}(1-z)^{\mu/2}Y_1(z)e_2$, where $z^{-\mu/2}$ is holomorphic on the unit disc centered at $z=1$. Thus, by means of \eqref{FloqSol1},
\begin{align*}
Y(z)e_2 
& = (1-z)^\mu z^{-\mu/2}\sum_{k=0}^\infty (1-z)^k b_k \\
& = (1-z)^\mu\sum_{k=0}^\infty (1-z)^k d_k^{(2)}\quad\mbox{with}\quad
d_0^{(2)} = b_0 = \begin{pmatrix} \frac{\beta+t}{\mu+1} \\[1ex] 1 \end{pmatrix}
\end{align*}
is a holomorphic function on $\mathfrak{D}_1$.

\begin{Lemma} \label{lem:ConCoeff}
Let $d_k\in\Cz$ be the series coefficients of the holomorphic solution \eqref{HolSol}, and
$\vartheta\in\Cz$ be given by \eqref{theta}. If $\re\mu > 0$ or $\mu=0$, then
\begin{equation} \label{LimCoeff}
\vartheta\T d_k = c_1 + \Osym(k^{\varepsilon-\mu-2})
\quad\mbox{as}\quad k\to\infty
\end{equation}
holds for any $\varepsilon>0$. In particular, $\lim_{k\to\infty}\vartheta\T d_k = c_1$.
\end{Lemma}

\begin{proof}
First, let us assume that $\re\mu>0$ holds and that $\mu$ is not an integer. In this case, $p_1=\frac{1}{\mu}$ and $q_1=0$ in \eqref{FundMat}, so that
\begin{align*}
\eta_1(z) & := Y(z)e_1 = \begin{pmatrix} \frac{1}{1-z} & 0 \\[1ex] 0 & 1 \end{pmatrix}H(z)\begin{pmatrix} 1 \\[1ex] 0 \end{pmatrix} = (1-z)^{-1}\sum_{k=0}^\infty (1-z)^k d_k^{(1)} \\
\eta_2(z) & := Y(z)e_2 = \begin{pmatrix} \frac{1}{1-z} & 0 \\[1ex] 0 & 1 \end{pmatrix}H(z)\begin{pmatrix} 0 \\[1ex] (1-z)^\mu \end{pmatrix} = (1-z)^{\mu}\sum_{k=0}^\infty (1-z)^k d_k^{(2)}
\end{align*}
is a fundamental set of Floquet solutions, where $d_0^{(1)} = e_1$. Now, using the results of Schäfke and Schmidt given in \cite{SS:1980} and \cite{Schaefke:1980}, we obtain a relationship between the series coefficients of $\eta_0$ and $\eta_1$, $\eta_2$ involving the connection coefficients $c_1$, $c_2$. From \cite[Theorem 1.4]{SS:1980} with $\alpha=0$, $\alpha_1=-1$, $\alpha_2=\mu$ it follows that
\begin{equation*}
d_k 
= c_1\sum_{\ell=0}^{n_1}\frac{\Gamma(k-\ell+1)}{\Gamma(k+1)\Gamma(-\ell+1)}d_\ell^{(1)} 
+ c_2\sum_{\ell=0}^{n_2}\frac{\Gamma(k-\ell-\mu)}{\Gamma(k+1)\Gamma(-\ell-\mu)}d_\ell^{(2)} 
+ \Osym(k^{-\nu-1})
\end{equation*}
holds for $k\to\infty$ with arbitrary integers $n_1,n_2\geq 0$ and $\nu := \min\{n_1,\mu+n_2+1\}$. Since $\frac{1}{\Gamma(-\ell+1)}=0$ for any positive integer $\ell$, we get the formula
\begin{equation*}
d_k = c_1\frac{\Gamma(k+1)}{\Gamma(k+1)\Gamma(1)}d_0^{(1)} 
+ c_2\sum_{\ell=0}^{n_2}\frac{\Gamma(k-\ell-\mu)}{\Gamma(k+1)\Gamma(-\ell-\mu)}d_\ell^{(2)} 
+ \Osym(k^{-\nu-1})
\end{equation*}
which is independent of $n_1$. If we choose $n_1$ sufficiently large, then $\nu = \mu+n_2+1$. Moreover, if we set $n_2=0$, then $\nu=\mu+1$ and
\begin{equation*}
d_k = c_1 e_1 + c_2\frac{\Gamma(k-\mu)}{\Gamma(k+1)\Gamma(-\mu)}b_0 + \Osym(k^{-\mu-2})
\end{equation*}
where we have inserted $d_0^{(1)}=e_1$, $d_0^{(2)}=b_0$ in a final step. Note that the vectors
\begin{equation*}
\vartheta = \begin{pmatrix} \ms 1 \\[1ex] -\frac{\beta+t}{\mu+1} \end{pmatrix}\quad\mbox{and}\quad
b_0 =  \begin{pmatrix} \frac{\beta+t}{\mu+1} \\[1ex] 1 \end{pmatrix}
\end{equation*}
are orthogonal, i.e., $\vartheta\T b_0 = 0$, and in addition $\vartheta\T e_1=1$. Multiplying above asymptotic expansion for $d_k$ by $\vartheta\T$ from the left, we obtain \eqref{LimCoeff} for $\varepsilon=0$ and thus for any $\varepsilon>0$. 

It remains to consider the case where $\mu=m$ is a non-negative integer. Here, \eqref{FundMat} takes the form
\begin{align*}
Y(z) 
& = \begin{pmatrix} \frac{1}{1-z} & 0 \\[1ex] 0 & 1 \end{pmatrix}H(z)
\begin{pmatrix} 1 & 0 \\[1ex] 0 & (1-z)^m \end{pmatrix}
\begin{pmatrix} 1 & 0 \\[1ex] q_1\log(1-z) & 1 \end{pmatrix} \\
& = (1-z)^{-1}\hat H(z)\begin{pmatrix} 1 & 0 \\[1ex] q_1\log(1-z) & 1 \end{pmatrix}
\end{align*} 
where
\begin{equation*}
\hat H(z) := \begin{pmatrix} 1 & 0 \\[1ex] 0 & 1-z \end{pmatrix}H(z)
\begin{pmatrix} 1 & 0 \\[1ex] 0 & (1-z)^m \end{pmatrix} = \sum_{k=0}^\infty (1-z)^k D_k
\end{equation*}
is a holomorphic matrix function on $\mathfrak{D}_1$ with coefficients
\begin{equation} \label{AsymExp}
D_0 = \begin{pmatrix} 1 & 0 \\[1ex] 0 & 0 \end{pmatrix},\quad
D_k = \begin{pmatrix} \ast & 0 \\[1ex] \ast & 0 \end{pmatrix}\quad\mbox{for $k=1,\ldots,m$,}\quad D_{m+1} = \begin{pmatrix} \ast & \ast \\[1ex] \ast & 1 \end{pmatrix}
\end{equation}
In particular, $D_k e_2 = o$ for $k=0,\ldots,m$ is the zero vector, and therefore
\begin{equation*}
\eta_2(z) 
= (1-z)^{-1}\sum_{k=m+1}^\infty (1-z)^k D_k\begin{pmatrix} 0 \\[1ex] 1 \end{pmatrix}
= (1-z)^m\sum_{k=0}^\infty (1-z)^k D_{k+m+1}\begin{pmatrix} 0 \\[1ex] 1 \end{pmatrix}
\end{equation*} 
is a Floquet solution of \eqref{RegSing0} corresponding to the characteristic exponent $m$ at $z=1$. Hence, 
\begin{equation*}
D_{m+1} e_2 = b_0 = \begin{pmatrix} \frac{\beta+t}{\mu+1} \\[1ex] 1 \end{pmatrix},\quad\mbox{i.e.,}\quad
D_{m+1} = \begin{pmatrix} \ast & \frac{\beta+t}{\mu+1} \\[1ex] \ast & 1 \end{pmatrix}
\end{equation*}
Moreover, since
\begin{equation*}
(1-z)^{-1}\begin{pmatrix} 1 & 0 \\[1ex] q_1\log(1-z) & 1 \end{pmatrix} = (1-z)^Q
\quad\mbox{with}\quad Q := \begin{pmatrix} -1 & \ms 0 \\[1ex] q_1 & -1 \end{pmatrix}
\end{equation*}
the fundamental matrix \eqref{FundMat} becomes $Y(z) = \sum_{k=0}^\infty (1-z)^k D_k(1-z)^Q$. Applying \cite[Theorem 1.1]{Schaefke:1980} with $\alpha=0$ and $\gamma_{-}=-1$, we obtain
\begin{equation} \label{dCoeff}
d_k = \sum_{\ell=0}^{n}\frac{1}{\Gamma(k+1)}D_\ell\frac{1}{\Gamma}(-\ell-Q)\Gamma(k-\ell-Q)c + \Osym\left(k^{\varepsilon-n-1}\right)
\end{equation}
for $k\to\infty$ with arbitrary $n\in\N$ and $\varepsilon > 0$. The definition and properties of the reciprocal Gamma function for matrices can be found in the appendix of \cite{Schaefke:1980}. In particular, for a lower-triangular Jordan block we have
\begin{equation*}
\frac{1}{\Gamma}\begin{pmatrix} \kappa & 0 \\[1ex] q & \kappa \end{pmatrix}  
= \begin{pmatrix} \frac{1}{\Gamma}(\kappa) & 0 \\[1ex] q\,(\frac{1}{\Gamma})'(\kappa) & \frac{1}{\Gamma}(\kappa) \end{pmatrix}\quad\mbox{for any}\quad q,\kappa\in\C
\end{equation*}
according to \cite[Theorem A.2.(ii)]{Schaefke:1980}, which implies
\begin{equation*}
D_0\frac{1}{\Gamma}(-Q)\Gamma(k-Q) =
\begin{pmatrix} 1 & 0 \\[1ex] 0 & 0 \end{pmatrix}
\begin{pmatrix} 1 & 0 \\[1ex] \ast & 1 \end{pmatrix}
\begin{pmatrix} \Gamma(k+1) & 0 \\[1ex] \ast & \Gamma(k+1) \end{pmatrix} = \Gamma(k+1)D_0
\end{equation*}
Furthermore, for any integers $\ell>0$ and $k>\ell-1$ we get
\begin{align*}
& D_\ell\frac{1}{\Gamma}(-\ell-Q)\Gamma(k-\ell-Q) \\
& = D_\ell\begin{pmatrix} \frac{1}{\Gamma}(1-\ell) & 0 \\[1ex] q_1(\frac{1}{\Gamma})'(1-\ell) & \frac{1}{\Gamma}(1-\ell) \end{pmatrix}\begin{pmatrix} \Gamma(k-\ell+1) & 0 \\[1ex] \ast & \Gamma(k-\ell+1) \end{pmatrix} \\
& = D_\ell\begin{pmatrix} 0 & 0 \\[1ex] (-1)^{\ell-1}(\ell-1)!\,q_1\Gamma(k-\ell+1) & 0 \end{pmatrix}
\end{align*}
where $\frac{1}{\Gamma}(1-\ell)=0$ and $(\frac{1}{\Gamma})'(1-\ell)=(-1)^{\ell-1}(\ell-1)!$ have been used. In particular, \eqref{AsymExp} yields $D_\ell\frac{1}{\Gamma}(-\ell-Q)\Gamma(k-\ell-Q)=O$ for $\ell = 1,\ldots,m$, and in case of $\ell = m+1$ we receive
\begin{equation*}
D_{m+1}\frac{1}{\Gamma}(-m-1-Q)\Gamma(k-m-1-Q)
 = (-1)^{m}m!\,q_1\Gamma(k-m)D_{m+1}\begin{pmatrix} 0 & 0 \\[1ex] 1 & 0 \end{pmatrix}
\end{equation*}
If we apply these results to \eqref{dCoeff} with $n=m+1$, we obtain
\begin{align*}
d_k 
& = \sum_{\ell=0}^{m+1}\frac{1}{\Gamma(k+1)}D_\ell\frac{1}{\Gamma}(-\ell-Q)\Gamma(k-\ell-Q)c + \Osym\left(k^{\varepsilon-m-2}\right) \\
& = D_0 c + (-1)^{m}m!\,q_1\frac{\Gamma(k-m)}{\Gamma(k+1)}D_{m+1}\begin{pmatrix} 0 & 0 \\[1ex] 1 & 0 \end{pmatrix}c + \Osym\left(k^{\varepsilon-m-2}\right)
\end{align*} 
or, by simplification,
\begin{equation*}
d_k = c_1 e_1 + c_1(-1)^{m}m!\,q_1 k^{-m-1}b_0  + \Osym\left(k^{\varepsilon-m-2}\right)
\end{equation*}
Multiplying this asymptotic expansion from the left by the vector $\vartheta\T$, which is orthogonal to $b_0$, we end up with \eqref{LimCoeff}.
\end{proof}

In order to get the connection coefficient $c_1=\lim_{k\to\infty}\vartheta\T d_k$, we have to determine the coefficients $d_k$ of the holomorphic function $\eta_0(z)=\sum_{k=0}^\infty d_k z^k$ 
satisfying $\eta_0'(z) = \big(\frac{1}{z}A_0 + \frac{1}{z-1}A_1 + C\big)\eta_0(z)$, where the matrices $A_0$, $A_1$, $C$ are defined in \eqref{ResMat} and the vector $d_0$ is given by \eqref{HolSol}. If we introduce
\begin{equation*}
B := A_0+A_1-C+E = \begin{pmatrix} -\mu-1 & 2\beta+4\gamma^2 \\[1ex] -1 & \mu+1 \end{pmatrix},\quad
d_{-1} := \begin{pmatrix} 0 \\[1ex] 0 \end{pmatrix}
\end{equation*}
then $d_k$ can be uniquely determined by the recurrence relation
\begin{equation} \label{RecForm}
d_k = (A_0-k)^{-1}\left((B-k)d_{k-1} + C d_{k-2}\right)\quad\mbox{for}\quad k = 1,2,3,\ldots
\end{equation}
Since $B$, $C$ are independent of $t$ and
\begin{equation*}
(A_0-k)^{-1} = \begin{pmatrix} -\frac{1}{\mu+k+1} & \frac{t-\beta}{k(\mu+k+1)} \\[1ex] 0 & -\frac{1}{k} \end{pmatrix}
\end{equation*}
the components of the vector $d_k=d_k(t)$ are polynomials in $t$. To obtain their leading coefficients, we need to evaluate the recursion formula \eqref{RecForm} to some extend. For this reason, we denote by $\Poly_n(t)$ an arbitrary polynomial in $t$ of degree equal or less than $n$, and we set $\Poly_n(t)\equiv 0$ in the case $n<0$. By induction we will prove that
\begin{equation} \label{Polynom}
d_k = \begin{pmatrix} r_k t^{k+1} + \Poly_k(t) \\[1ex] s_k t^k + \Poly_{k-1}(t) \end{pmatrix}
\end{equation}
for all non-negative integers $k$. In fact, this is true for $k=0$, since
\begin{equation*}
d_0 = \begin{pmatrix} \frac{\beta-t}{\mu+1} \\[1ex] 1 \end{pmatrix}
= \begin{pmatrix} r_0 t + \Poly_0(t) \\[1ex] 1 \end{pmatrix}
\quad\mbox{with}\quad r_0 = -\tfrac{1}{\mu+1}\quad\mbox{and}\quad s_0 = 1
\end{equation*}
Further, assuming that the vectors $d_{k-1}$, $d_{k-2}$ in the recurrence relation \eqref{RecForm}
already have the form \eqref{Polynom} and taking into account, that the term $C d_{k-2}$ is absent in the case $k=1$, then we get for $k>0$
\begin{align*}
d_k & = \begin{pmatrix} -\frac{1}{\mu+k+1} & \frac{t-\beta}{k(\mu+k+1)} \\[1ex] 0 & -\frac{1}{k} \end{pmatrix}\left(\begin{pmatrix} -\mu-1-k & 2\beta+4\gamma^2 \\[1ex] -1 & \mu+1-k \end{pmatrix}
\begin{pmatrix} r_{k-1}\,t^k + \Poly_{k-1}(t) \\[1ex] s_{k-1}\,t^{k-1} + \Poly_{k-2}(t) \end{pmatrix} + {}\right. \\ & \quad \left. {} + \begin{pmatrix} 0 & -4\gamma^2 \\[1ex] 1 & 0 \end{pmatrix}\begin{pmatrix} r_{k-2}\,t^{k-1} + \Poly_{k-2}(t) \\[1ex] s_{k-2}\,t^{k-2} + \Poly_{k-3}(t) \end{pmatrix}\right) 
= \begin{pmatrix} r_k t^{k+1} + \Poly_{k-1}(t) \\[1ex] s_k t^k + \Poly_{k-1}(t) \end{pmatrix}
\end{align*}
where $r_k = -\frac{1}{k(\mu+k+1)}\,r_{k-1}$ and $s_k = \frac{1}{k}\,r_{k-1}=-(\mu+k+1)r_k$. Starting with $r_0 = -\tfrac{1}{\mu+1}$, we obtain
\begin{equation*}
r_k = \frac{(-1)^{k+1}}{k!\,(\mu+1)(\mu+2)\cdots(\mu+k+1)},\quad k=0,1,2,3,\ldots
\end{equation*}
If we multiply $d_k$ from the left by $\vartheta\T$, then we get
\begin{equation*}
\Theta_k(t) = r_k t^{k+1} - \tfrac{s_k}{\mu+1}\,t^{k+1} + \Poly_k(t)
= \left(2+\tfrac{k}{\mu+1}\right)r_k t^{k+1} + \mbox{(lower order terms in $t$)}
\end{equation*}
so that $\Theta_k(t)$ is a polynomial of degree $k+1$ in $t$. Now, according to \cref{lem:ConCoeff}, $\Theta(t) := \lim_{k\to\infty}\Theta_k(t)$ coincides with the connection coefficient $c_1$ between $\eta_0(z)$ and $\eta_1(z)$, and therefore $\lambda=t+\mu(\mu+1)$ is an eigenvalue of \eqref{CSWE} if and only if $\Theta(t)=0$. Hence, the eigenvalues of \eqref{CSWE} are exactly the zeros of the entire function $\Theta$ up to the translation by $\mu(\mu+1)$. Finally, in order to simplify the recurrence relation, we rearrange \eqref{RecForm} as follows:
\begin{align*}
d_k 
& = (A_0-k)^{-1}\big((A_0-k+A_1-C+E)d_{k-1} + C d_{k-2}\big) \\
& = d_{k-1} + (A_0-k)^{-1}\big((A_1+E)d_{k-1} - C(d_{k-1}-d_{k-2})\big)
\end{align*}
If we set $u_k := d_k-d_{k-1}$, then this expression becomes
\begin{align*}
u_k & = (A_0-k)^{-1}\big((A_1+E)d_{k-1} - Cu_{k-1}\big) \\ 
d_k & = d_{k-1}+u_k,\quad k=1,2,3,\ldots
\end{align*}
starting with $u_0=d_0$. Note that
\begin{equation*}
(A_0-k)^{-1}(A_1+E) 
= \begin{pmatrix} 0 & \frac{t-\beta}{k}-\frac{2t}{k+\mu+1} \\[1ex] 
  0 & -\frac{\mu+1}{k} \end{pmatrix},\ \ 
(A_0-k)^{-1}C 
= \begin{pmatrix} \frac{t-\beta}{k(k+\mu+1)} & \frac{4\gamma^2}{k+\mu+1} \\[1ex]
  -\frac{1}{k} & 0 \end{pmatrix}
\end{equation*}
This completes the proof of assertions (a) -- (c) in \cref{thm:MainRes}. It remains to verify the explicit expression in (d) for $\beta=\gamma=0$ and $t=(\nu+\mu+1)(\nu-\mu)$. In this case $\lambda=t+\mu(\mu+1) = \nu(\nu+1)$, and \eqref{CSWE} simply becomes the associated Legendre differential equation
\begin{equation} \label{Legendre}
\D{}{x}\left((1-x^2)\D{}{x}w(x)\right) + \left(\nu(\nu+1) - \frac{\mu^2}{1-x^2}\right)w(x) = 0
\end{equation}
Let us initially assume $\nu\pm\mu\neq -1,-2,-3,\ldots$ in addition to $\re\mu>0$ or $\mu=0$. We can then give two fundamental solutions of \eqref{Legendre}, namely the associated Legendre function of the first kind with the asymptotic behavior
\begin{equation*}
P_{\nu}^{-\mu}(-x) \sim \tfrac{1}{\Gamma(1+\mu)}\left(\tfrac{1+x}{2}\right)^{\mu/2}
\quad\mbox{for}\quad x\to -1
\end{equation*}
and the associated Legendre function of the second kind $Q_{\nu}^{-\mu}(-x)$, which is not bounded as $x\to -1$ (cf. \cite[Section 4.8.2]{MOS:1966}). If we set $z=\frac{1+x}{2}$, then the second component of the vector function $y_0(z)$ given by \eqref{FloqSol0} is a solution of \eqref{Legendre} which is bounded near $x=-1$. Hence, $y_0(z)$ takes the form
\begin{equation*}
y_0(z) = \left(\tfrac{1+x}{2}\right)^{\mu/2}\begin{pmatrix} \ast \\[1ex] 1+o(1) \end{pmatrix}
= \begin{pmatrix} \ast \\[1ex] \Gamma(\mu+1)P_{\nu}^{-\mu}(-x) \end{pmatrix}
\end{equation*}
The associated Legendre functions $P_{\nu}^{-\mu}(x)$, $Q_{\nu}^{-\mu}(x)$ with the asymptotic behavior
\begin{equation*}
P_{\nu}^{-\mu}(x) \sim \tfrac{1}{\Gamma(1+\mu)}\left(\tfrac{1-x}{2}\right)^{\mu/2},\quad
Q_{\nu}^{-\mu}(x) \sim 
\left\{\begin{array}{ll}
\tfrac{\Gamma(\mu)\Gamma(\nu-\mu+1)}{2\,\Gamma(\nu+\mu+1)}\left(\tfrac{1-x}{2}\right)^{-\mu/2},
& \re\mu>0 \\[2ex]
-\tfrac{1}{2}(1+o(1))\log\left(\tfrac{1-x}{2}\right),
& \mu=0\end{array}\right.
\end{equation*}
as $x\to 1$ form another fundamental basis of \eqref{Legendre}. Moreover, the second component of $y_1(x)$ given by \eqref{FundSol1} is a solution of \eqref{Legendre}, where \eqref{Asymp1} implies
\begin{equation*}
y_1(z) = \begin{pmatrix} \ast \\[1ex] 
\left(\frac{1-x}{2}\right)^{-\mu/2}(p_1+o(1)) + q_1(1+o(1))\left(\frac{1-x}{2}\right)^{\mu/2}\log\left(\frac{1-x}{2}\right) \end{pmatrix}
\end{equation*}
as $x\to 1$, where $p_1=\frac{1}{\mu}$ if $\re\mu>0$ and $p_1=0$, $q_1=-1$ in the case $\mu=0$. Therefore, if $\re\mu>0$, then above expression for $y_1(z)$ in combination with $\mu\,\Gamma(\mu)=\Gamma(\mu+1)$ yields
\begin{equation} \label{Connect}
y_1(z) = \begin{pmatrix} \ast \\[1ex] \frac{2\,\Gamma(\nu+\mu+1)}{\Gamma(\mu+1)\Gamma(\nu-\mu+1)}Q_{\nu}^{-\mu}(x) + \omega_1 P_{\nu}^{-\mu}(x) \end{pmatrix}
\end{equation}
with some constant $\omega_1$. In case of $\mu=0$ it follows that
\begin{equation*}
y_1(z) 
= \begin{pmatrix} \ast \\[1ex] 2\,Q_{\nu}^0(x) + \omega_1 P_{\nu}^0(x) \end{pmatrix}
= \begin{pmatrix} \ast \\[1ex] \frac{2\,\Gamma(\nu+0+1)}{\Gamma(0+1)\Gamma(\nu-0+1)}\,Q_{\nu}^0(x) + \omega_1 P_{\nu}^0(x) \end{pmatrix}
\end{equation*}  
and hence \eqref{Connect} is also valid for $\mu=0$. Furthermore, the second component of $y_2(z)$ given by \eqref{FundSol1} is bounded near $x=1$, and therefore
\begin{equation*}
y_2(z) = \begin{pmatrix} \ast \\[1ex] \omega_2 P_{\nu}^{-\mu}(x) \end{pmatrix}
\end{equation*}  
with some constant $\omega_2$. Now, if we write $y_0(z) = c_1 y_1(z) + c_2 y_2(z)$ with the connection coefficients $c_1$ and $c_2$, then we get
\begin{equation} \label{Connect1}
\begin{pmatrix} \ast \\[1ex] \Gamma(\mu+1)P_{\nu}^{-\mu}(-x) \end{pmatrix}
= \begin{pmatrix} \ast \\[1ex] \frac{2\,\Gamma(\nu+\mu+1)c_1}{\Gamma(\mu+1)\Gamma(\nu-\mu+1)}Q_{\nu}^{-\mu}(x) + (c_1\omega_1+c_2\omega_2)P_{\nu}^{-\mu}(x) \end{pmatrix}
\end{equation}
On the other hand, according to the connection formula \cite[14.9.10]{NIST:2010}, we have
\begin{equation} \label{Connect2}
P_{\nu}^{-\mu}(-x) = \cos\left((\nu-\mu)\pi\right)P_{\nu}^{-\mu}(x) - \tfrac{2}{\pi}\sin\left((\nu-\mu)\pi\right)Q_{\nu}^{-\mu}(x)
\end{equation}
Comparing \eqref{Connect1} to \eqref{Connect2} yields $-\tfrac{2}{\pi}\,\Gamma(\mu+1)\sin\left((\nu-\mu)\pi\right) = \frac{2\,\Gamma(\nu+\mu+1)c_1}{\Gamma(\mu+1)\Gamma(\nu-\mu+1)}$ and therefore
\begin{equation*}
\Theta(t) = c_1(t) = \frac{\sin\left((\mu-\nu)\pi\right)\Gamma(\mu+1)^2\Gamma(\nu-\mu+1)}{\pi\Gamma(\nu+\mu+1)}
\end{equation*}
If we apply the functional relation
$\frac{\sin(\pi z)}{\pi}=\frac{1}{\Gamma(z)\Gamma(1-z)}$ with $z=\mu-\nu$, then we get
\begin{equation*}
\Theta(t) = \frac{\Gamma(\mu+1)^2}{\Gamma(\mu-\nu)\Gamma(\mu+\nu+1)},\quad\mbox{where}\quad 
t=(\nu+\mu+1)(\nu-\mu)
\end{equation*}
We have proved this relation under the additional assumption $\nu\pm\mu\neq -1,-2,-3,\ldots$, but since the expressions on both sides are entire functions, it is valid for all $\nu\in\C$.
Finally, if we set $\tau := \sqrt{t+(\mu+\tfrac{1}{2})^2}$, where $-\frac{\pi}{2}<\arg(\tau)\leq\frac{\pi}{2}$, then $\nu = -\frac{1}{2}+\tau$ and
\begin{equation*}
\Theta(t) 
= \frac{\cos\left((\tau-\mu)\pi\right)\Gamma(\mu+1)^2\Gamma(\tau+\frac{1}{2}-\mu)}{\pi\Gamma(\tau+\frac{1}{2}+\mu)} = \frac{\Gamma(\mu+1)^2}{\Gamma(\mu+\frac{1}{2}+\tau)\Gamma(\mu+\frac{1}{2}-\tau)}
\end{equation*}
which completes the proof of \cref{thm:MainRes}.
  
\section{Notes on the generalized equation} 
\label{sec:GSWE} 

In this paper a new approach for the computation of Coulomb spheroidal eigenvalues has been presented: The eigenvalues are the zeros of a holomorphic function that can be obtained with a comparatively simple recurrence procedure which also provides the coefficients for the series expansion of the corresponding Coulomb spheroidal wave functions. Of course, all results can be applied to the angular spheroidal wave equation as well. Even more, the method introduced here may also be used for computing the eigenvalues of the generalized spheroidal wave equation \eqref{GSWE}. By means of
\begin{equation*}
v(x) := 2w'(x) + \frac{2(\mu x + \alpha)}{1-x^2}\,w(x),\quad
y(z) = \begin{pmatrix} v(2z-1) \\[1ex] w(2z-1) \end{pmatrix}
\end{equation*}
the GSWE is equivalent to the $2\times 2$ system
\begin{equation} \label{ExtSys}
y'(z) = \left(\frac{1}{z}\begin{pmatrix} -\frac{\mu-\alpha}{2}-1 & \beta-t \\[1ex] 0 & \frac{\mu-\alpha}{2} \end{pmatrix} + \frac{1}{z-1}\begin{pmatrix} -\frac{\mu+\alpha}{2}-1 & \beta+t \\[1ex] 0 & \frac{\mu+\alpha}{2} \end{pmatrix} + \begin{pmatrix} 0 & -4\gamma^2 \\[1ex] 1 & 0 \end{pmatrix}\right)y(z)
\end{equation}
where $t := \lambda-\mu(\mu+1)$; it coincides with \eqref{RegSing} except for the diagonal entries and the Floquet exponents, respectively. Since the GSWE remains unchanged if we replace $\mu$, $\alpha$ by $-\mu$, $-\alpha$ or interchange $\mu$ and $\alpha$, we can assume without loss of generality that $\re(\mu-\alpha)\geq 0$ as well as $\re(\mu+\alpha)\geq 0$ holds. Moreover, if we substitute $-x$ for $x$ in \eqref{CSWE}, then we get a GSWE with parameters $-\alpha$, $-\beta$ instead of $\alpha$, $\beta$. Therefore, we may further suppose without restriction that $\re\alpha\geq 0$. Under these assumptions, the considerations in \cref{sec:Proof} concerning the Floquet solutions and their connection coefficients remain valid for the system \eqref{ExtSys}, requiring only minor formal adjustments. In the following, we will briefly sketch the key steps, and for convenience we will additionally assume $\re(\mu\pm\alpha)>0$ and $\mu\pm\alpha\not\in\Z$.

If we apply \cref{lem:StructFM} to both \eqref{ExtSys} and the system transformed with $\tilde y(z) := y(1-z)$, then we obtain two fundamental matrices
\begin{align*}
Y_0(z) & = \begin{pmatrix} \frac{1}{z} & 0 \\[1ex] 0 & 1 \end{pmatrix}H_0(z)
\begin{pmatrix} z^{-(\mu-\alpha)/2} & 0 \\[1ex] 0 & z^{(\mu-\alpha)/2} \end{pmatrix} \\
Y_1(z) & = \begin{pmatrix} \frac{1}{1-z} & 0 \\[1ex] 0 & 1 \end{pmatrix}H_1(z)
\begin{pmatrix} (1-z)^{-(\mu+\alpha)/2} & 0 \\[1ex] 0 & (1-z)^{(\mu+\alpha)/2} \end{pmatrix}
\end{align*}
where $H_a(z)$ are holomorphic matrix functions on $\mathfrak{D}_a := \{z\in\C:|z-a|<1\}$ for $a\in\{0,1\}$ satisfying
\begin{equation*}
H_0(0) = \begin{pmatrix} 1 & 0 \\[1ex] \frac{1}{\alpha-\mu} & 1 \end{pmatrix}
\quad\mbox{and}\quad
H_1(1) = \begin{pmatrix} 1 & 0 \\[1ex] \frac{1}{\mu+\alpha} & 1 \end{pmatrix}
\end{equation*}
Similar to \cref{lem:RegSol0}, $\lambda = t+\mu(\mu+1)$ is an eigenvalue of \eqref{GSWE} if and only if
\begin{align*}
y_0(z) & = Y_0(z)e_2 = z^{(\mu-\alpha)/2}\sum_{k=0}^\infty z^k a_k,\quad
a_0 := \begin{pmatrix} \frac{\beta-t}{\mu-\alpha+1} \\[1ex] 1 \end{pmatrix} \\
Y_1(z)e_2 & = (1-z)^{(\mu+\alpha)/2}\sum_{k=0}^\infty (1-z)^k b_k,\quad
b_0 := \begin{pmatrix} \frac{\beta+t}{\mu+\alpha+1} \\[1ex] 1 \end{pmatrix}
\end{align*}
are constant multiples of each other. On the other hand, we have
\begin{equation*}
y_0(z) = Y_1(z)\begin{pmatrix} c_1(t) \\[1ex] c_2(t) \end{pmatrix}
\end{equation*}
where the connection coefficients $c_1=c_1(t)$, $c_2=c_2(t)$ depend holomorphically on $t\in\C$. Hence, $\lambda=t+\mu(\mu+1)$ is an eigenvalue of \eqref{GSWE} if and only if $c_1(t)=0$.
To obtain a simple calculation formula for $c_1(t)$, we apply the transformation $\eta(z) = z^{-(\mu-\alpha)/2}(1-z)^{(\mu+\alpha)/2}y(z)$. The resulting system
\begin{equation*}
\eta'(z) = \left(\frac{1}{z}\begin{pmatrix} -\mu+\alpha-1 & \beta-t \\[1ex] 0 & 0 \end{pmatrix} + \frac{1}{z-1}\begin{pmatrix} -1 & \beta+t \\[1ex] \ms 0 & \mu+\alpha \end{pmatrix} + \begin{pmatrix} 0 & -4\gamma^2 \\[1ex] 1 & 0 \end{pmatrix}\right)\eta(z)
\end{equation*}
has a holomorphic solution $\eta_0(z) = \sum_{k=0}^\infty z^k d_k$ on $\mathfrak{D}_0$ with $d_0 = a_0$ and two fundamental solutions
\begin{equation*}
\eta_1(z) =  (1-z)^{-1}\sum_{k=0}^\infty (1-z)^k d_k^{(1)},\quad
\eta_2(z) = (1-z)^{\mu+\alpha}\sum_{k=0}^\infty (1-z)^k d_k^{(2)}
\end{equation*}
on $\mathfrak{D}_1$, where $d_k^{(1)} = e_1$ and $d_k^{(2)} = b_0$; they are connected by $\eta_0(z) = c_1\eta_1(z) + c_2\eta_2(z)$. Now, as in \cref{lem:ConCoeff} we obtain
\begin{equation*}
d_k = c_1 e_1 + c_2\frac{\Gamma(k-\mu-\alpha)}{\Gamma(k+1)\Gamma(-\mu-\alpha)}b_0 + \Osym(k^{-\mu-\alpha-2})
\end{equation*}
Further, if we define
\begin{equation*}
\vartheta = \begin{pmatrix} \ms 1 \\[1ex] -\frac{\beta+t}{\mu+\alpha+1} \end{pmatrix}
\end{equation*}
then $\vartheta\T e_1=1$ and $\vartheta\T b_0 = 0$. Hence, $\vartheta\T d_k = c_1(t) + \Osym(k^{-\mu-\alpha-2})$ as $k\to\infty$. Finally, the zeros $t_n$ of the entire function $\Theta(t) := c_1(t) = \lim_{k\to\infty}\vartheta\T d_k$ yield the eigenvalues $\lambda_n=t_n+\mu(\mu+1)$ of \eqref{GSWE}. 

The function $\Theta(t)$ can be calculated by a procedure similar to that in \cref{cor:NumCom}: If we define recursively 
\begin{align*}
a_k & := \frac{(\beta-t)a_{k-1}}{k(k+\mu-\alpha+1)} - \frac{4\gamma^2 b_{k-1}}{k+\mu-\alpha+1} 
       + \frac{(t-\beta)(\mu+\alpha+1)-(t+\beta)k}{k(k+\mu-\alpha+1)}\,w_{k-1} \\
b_k & := \frac{1}{k}\,a_{k-1} - \frac{\mu+\alpha+1}{k}\,w_{k-1},\quad w_k := b_k + w_{k-1}, \quad
\Theta_k := \Theta_{k-1} + a_k - \frac{(\beta+t)b_k}{\mu+\alpha+1}
\end{align*}
for $k=1,2,3,\ldots$ starting with $a_0 =\frac{\beta-t}{\mu-\alpha+1}$, $b_0=w_0=1$,
$\Theta_0 = \frac{\beta-t}{\mu-\alpha+1}-\frac{\beta+t}{\mu+\alpha+1}$, then
\begin{equation*}
\Theta_k = \Theta(t) + \Osym(k^{-\mu-\alpha-2})
\end{equation*}
as $k\to\infty$. Now, $\lambda\in\C$ is an eigenvalue of the generalized spheroidal wave equation \eqref{GSWE} if and only if $t = \lambda-\mu(\mu+1)$ is a zero of $\Theta$, and in this case the corresponding eigenfunctions are constant multiples of
\begin{equation*}
w(x) := \frac{(1+x)^{(\mu-\alpha)/2}}{(1-x)^{(\mu+\alpha)/2}}\sum_{k=0}^\infty\frac{w_k}{2^k}\,(1+x)^k
\end{equation*}
\Cref{fig:GSWF} shows some eigenfunctions of  \eqref{GSWE} for a sample set of parameter values.

\begin{figure}[tbhp]
\centering
\includegraphics{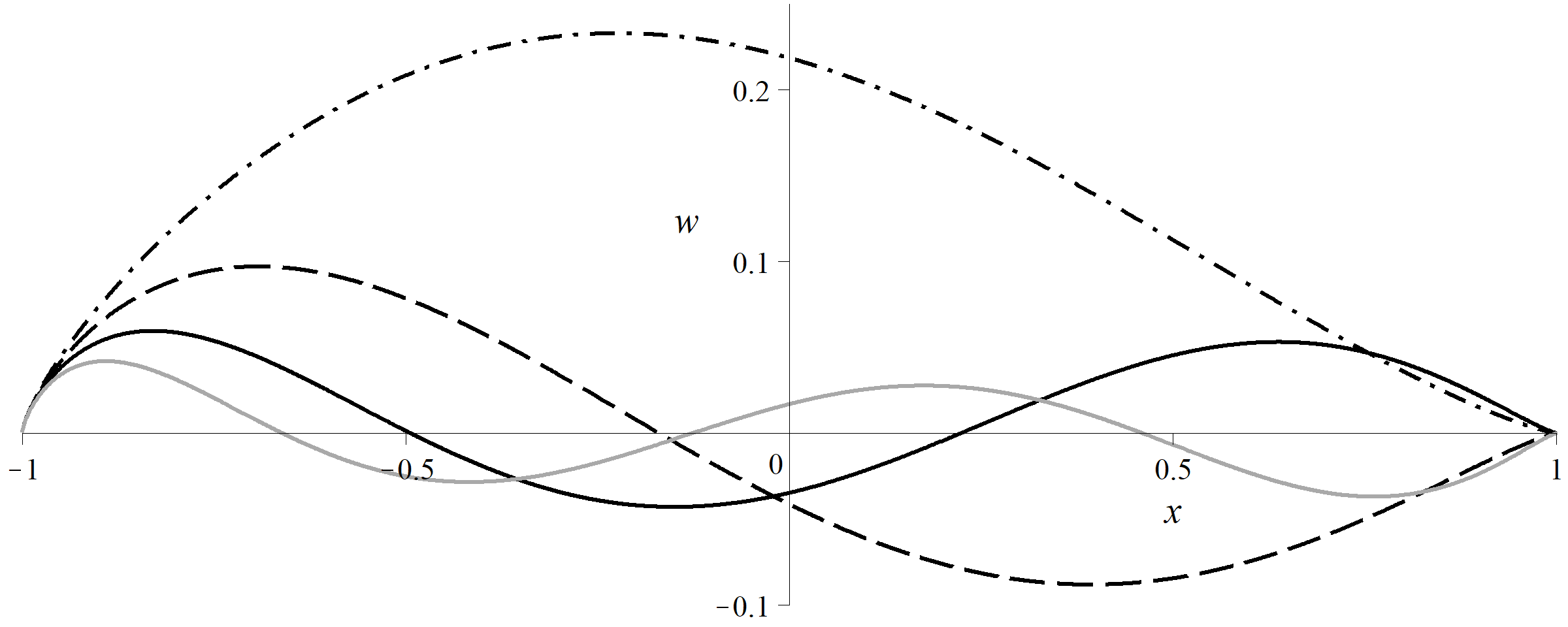}
\caption{Eigenfunctions $w(x)$ for the generalized spheroidal wave equation \eqref{GSWE} with parameters $\alpha=\frac{1}{2}$, $\beta=-1$, $\gamma=2$ and $\mu=2$. They are associated to the four lowest eigenvalues $\lambda_1 = 2.472312$, $\lambda_2 = 9.211599$, $\lambda_3 = 17.539555$, $\lambda_4 = 27.700922$ (each rounded to six decimal places).}\label{fig:GSWF}
\end{figure}

In our investigation of the GSWE, we have focused on parameter values for which the condition $\mu\pm\alpha\not\in\Z$ is satisfied in addition to $\re(\mu\pm\alpha)>0$. Basically, the algorithm described above still works if $\mu+\alpha\in\Z$ or $\mu-\alpha\in\Z$, but we will not discuss these special cases in more detail.

\end{document}